\documentclass[11pt,reqno]{amsart}
\usepackage{fullpage,hyperref}

\newtheorem{theorem}{Theorem}[section]
\newtheorem{lemma}[theorem]{Lemma}
\newtheorem{prop}[theorem]{Proposition}
\newtheorem{cor}[theorem]{Corollary}

\usepackage{graphicx}
\usepackage{color}
\usepackage{subfigure}
\usepackage{amssymb}
\usepackage{amsmath,mathrsfs}
\usepackage{colonequals}
\usepackage{hyperref}
\newcommand{\Clust}{\mathrm{\bf Clust}}
\theoremstyle{definition}
\newtheorem{definition}[theorem]{Definition}

\renewcommand{\setminus}{\smallsetminus}
\renewcommand{\le}{\leqslant}
\renewcommand{\ge}{\geqslant}
\renewcommand{\leq}{\leqslant}
\renewcommand{\geq}{\geqslant}
\renewcommand{\subset}{\subseteq}
\renewcommand{\supset}{\supseteq}
\renewcommand{\epsilon}{\varepsilon}

\newcommand{\abs}[1]{\left|#1\right|}                   
\newcommand{\absf}[1]{|#1|}                             
\newcommand{\vnorm}[1]{\left|\left|#1\right|\right|}    
\newcommand{\vnormf}[1]{||#1||}                         

\newcommand{\Z}{\mathbb{Z}}                             
\newcommand{\N}{\mathbb{N}}
\newcommand{\E}{\mathbb{E}}

\newcommand{\R}{\mathbb{R}}

\newcommand{\figoneawidth}{.4\textwidth}                
\newcommand{\figonebwidth}{.41\textwidth}
\newcommand{\figtwowidth}{.4\textwidth}

\newcommand{\e}{\varepsilon}
\renewcommand{\epsilon}{\varepsilon}
\renewcommand{\colonequals}{=}
\begin{document}

\title{Solution of the propeller conjecture in $\R^3$}
\thanks{S.~H. and A.~J. were supported by NSF grant CCF-0832795 and NSF Graduate
Research Fellowship  DGE-0813964.  A.~N. was supported
by NSF grant CCF-0832795, BSF grant 2006009, and the
Packard Foundation. Part of this work was carried out when S.~H. and A.~N. were visiting the Quantitative Geometry program at MSRI}
\author{Steven Heilman}
\address{Courant Institute, New York University, New York NY 10012}
\email{heilman@cims.nyu.edu}
\author{Aukosh Jagannath}
\address{Courant Institute, New York University, New York NY 10012}
\email{asj260@nyu.edu}
\author{Assaf Naor}
\address{Courant Institute, New York University, New York NY 10012}
\email{naor@cims.nyu.edu}

\maketitle
\begin{abstract}
It is shown that every measurable partition $\{A_1,\ldots, A_k\}$ of $\R^3$ satisfies
\begin{equation}\label{eq:abs}
\sum_{i=1}^k\left\|\int_{A_i} xe^{-\frac12\|x\|_2^2}dx\right\|_2^2\le 9\pi^2.
\end{equation}
Let $\{P_1,P_2,P_3\}$ be the partition of $\R^2$ into $120^\circ$
 sectors centered at the origin. The bound~\eqref{eq:abs} is sharp, with equality holding if   $A_i=P_i\times \R$ for $i\in \{1,2,3\}$ and $A_i=\emptyset$ for $i\in \{4,\ldots,k\}$ (up to measure zero corrections, orthogonal transformations and renumbering of the sets $\{A_1,\ldots,A_k\}$). This settles positively the $3$-dimensional Propeller Conjecture of Khot and Naor (FOCS 2008). The proof of~\eqref{eq:abs} reduces the problem to a finite set of numerical
inequalities which are then verified with full rigor in a computer-assisted fashion. The main consequence (and motivation) of~\eqref{eq:abs} is complexity-theoretic: the Unique Games hardness threshold of the Kernel Clustering problem with $4\times 4$ centered and spherical hypothesis matrix equals $\frac{2\pi}{3}$.
\end{abstract}

\begin{figure}[h]\label{fig:manor}
\begin{center}\includegraphics[scale=2.7]{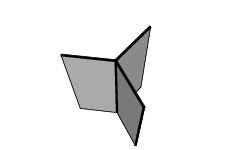}
\end{center}
 \caption{{{\em  The partition of $\R^3$ that maximizes the sum of squared lengths of Gaussian moments is a ``propeller": three planar $120^\circ$
 sectors multiplied by an orthogonal line, with the rest of the
   partition elements being empty. }}} \label{fig:midpoints}
\end{figure}

\newpage
\section{Introduction}\label{secintro}

For $x=(x_1,\ldots,x_m),y=(y_1,\ldots,y_m)\in \R^m$ let $\langle x,y\rangle =\sum_{i=1}^m x_i y_i$ denote their standard scalar product, and let $\|x\|_2=\sqrt{\langle x,x\rangle}$ denote the corresponding Euclidean norm. The cross product of $x=(x_1,x_2,x_3),y=(y_1,y_2,y_3)\in\R^3$ is denoted $x\times y=(x_2y_3-x_3y_2,x_3y_1-x_1y_3,x_1y_2-x_2y_1)$. The Gaussian measure on $\R^m$, i.e., the measure whose density is $x\mapsto (2\pi)^{-m/2}e^{-\|x\|_2^2/2}$, is denoted $\gamma_m$.

The following theorem is our main result, asserting that among all measurable partitions of $\R^3$, the ``propeller partition'' as depicted in Figure~\ref{fig:midpoints} maximizes the sum of the squared lengths of the Gaussian moments associated to the members of the partition.

\begin{theorem}[Main theorem; geometric formulation]\label{thm:main}
Let $\{A_1,\ldots,A_k\}$ be a partition of $\R^3$ into Lebesgue
measurable sets. For $i\in \{1,\ldots,k\}$ let $\zeta_i=\int_{A_i}
xd\gamma_3(x)\in \R^3$ be the Gaussian moment of the set $A_i$. Then
\begin{equation}\label{eq:z_i thm} \sum_{i=1}^k \|\zeta_i\|_2^2\le
\frac{9}{8\pi}.
\end{equation}
Let $\{P_1,P_2,P_3\}$ be the partition of $\R^2$ into $120^\circ$
 sectors centered at the origin. The bound~\eqref{eq:z_i thm} cannot be improved, with equality holding if   $A_i=P_i\times \R$ for $i\in \{1,2,3\}$ and $A_i=\emptyset$ for $i\in \{4,\ldots,k\}$ (up to measure zero corrections, orthogonal transformations, and renumbering of the sets $\{A_1,\ldots,A_k\}$).
\end{theorem}

The Propeller Conjecture was posed by Khot and Naor in~\cite{KN08}
as part of their investigation of the computational complexity of
the Kernel Clustering problem from Machine Learning. Specifically,
they conjectured the validity of the bound~\eqref{eq:z_i thm} for
measurable partitions of $\R^m$ for all $m>2$. They proved this
conjecture for $m=2$, and for $m=1$ they showed that the right hand side of~\eqref{eq:z_i thm} can be improved
to the sharp bound $\frac{1}{\pi}$. It is also shown in~\cite{KN08}
that proving~\eqref{eq:z_i thm} for $k=4$ implies the same
conclusion for all $k\in \N$. Correspondingly, for partitions of
$\R^m$ it suffices to prove this statement for $k=m+1$.

Before explaining the complexity-theoretic consequence of Theorem~\ref{thm:main}, which was the motivation of~\cite{KN08} for posing the Propeller Conjecture, we state two equivalent formulations of it: the first probabilistic (a sharp estimate for the expected maximum of a Gaussian vector) and the second analytic (a sharp Grothendieck inequality). The equivalence of these results to Theorem~\ref{thm:main} was established in~\cite{KN08}; for Theorem~\ref{thm:prob} below see~\cite[Lem.~3.7]{KN08} and  for Theorem~\ref{thm:gro} below see~\cite[Thm.~1.1]{KN08}.

\begin{theorem}[Main theorem; probabilistic formulation]\label{thm:prob} Let $(g_1,g_2,g_3,g_4)\in \R^4$ be a mean zero Gaussian vector (with arbitrary covariance matrix). then
\begin{equation}\label{eq:max gaussian}
\left|\E\left[\max_{i\in \{1,2,3,4\}}g_i\right]\right|\le
\frac{3}{2\sqrt{2\pi}}\sqrt{\sum_{i=1}^4\E\left[g_i^2\right]}.
\end{equation}
The bound~\eqref{eq:max gaussian} cannot be improved, with equality holding when the covariance matrix of $(g_1,g_2,g_3,g_4)$ equals
$$
\frac12\begin{pmatrix}
  2 & -1&-1&0 \\
   -1 & 2&-1&0\\
   -1&-1&2&0\\
   0&0&0&0
   \end{pmatrix}.
   $$
\end{theorem}

The following analytic formulation of Theorem~\ref{thm:main} is in terms of a sharp Grothendieck inequality. Grothendieck~\cite{Gro53} pioneered in 1953 the use of inequalities of this type in functional analysis, and ever since then such inequalities have permeated many mathematical disciplines. See~\cite{KN08,KN10,KN11} for an explanation of how the result below relates to a natural extension of the Grothendieck inequality and to approximation algorithms.

\begin{theorem}[Main theorem; analytic formulation]\label{thm:gro}
Let $(a_{ij})$  be an $n\times n$ positive semidefinite matrix with $\sum_{i=1}^{n}\sum_{j=1}^{n}a_{ij}=0$. For
every  $\{v_1,v_2,v_3,v_4\}\subseteq S^3$ with
$\sum_{i=1}^4v_i=0$ we have
\begin{equation}\label{eq:our gro}
\max_{x_1,\ldots,x_n\in S^{n-1}} \sum_{i=1}^n\sum_{j=1}^n a_{ij}
\langle x_i,x_j\rangle \le \frac{2\pi}{3}\max_{y_1,\ldots,y_n\in
\{v_1,v_2,v_3,v_4\}} \sum_{i=1}^n\sum_{j=1}^n a_{ij} \langle
y_i,y_j\rangle.
\end{equation}
The bound~\eqref{eq:our gro} cannot be improved, with asymptotic equality holding  when $\{v_1,v_2,v_3,v_4\}$ are the rows of the following matrix.
$$
\frac{1}{2\sqrt{3}}\begin{pmatrix}
  3 & -1&-1&-1 \\
   -1 & 3&-1&-1\\
   -1&-1&3&-1\\
   -1&-1&-1&3
   \end{pmatrix}.
   $$
See~\cite[Sec.~3.1]{KN10} for a description of a family of $n\times n$ matrices $(a_{ij})$  for which equality in~\eqref{eq:our gro} is asymptotically attained (as $n\to \infty$).
\end{theorem}
Inequality~\eqref{eq:our gro} was proved in~\cite{KN08}. Our new contribution is the assertion that~\eqref{eq:our gro} cannot be improved, a statement that is equivalent to inequality~\eqref{eq:z_i thm} of Theorem~\ref{thm:main}. This is not the first time that attempts to prove sharpness of a Grothendieck inequality led to an extremal geometric partitioning question. Notably, see K\"onig's conjecture~\cite{Kon00} as a step towards Krivine's conjecture~\cite{Krivine77} on the sharpness of his version of the classical Grothendieck inequality. Unlike the Propeller Conjecture in $\R^3$, these conjectures turned out to be false~\cite{BMMN11}, but they do indicate the interconnection between Grothendieck inequalities and extremal geometric partitioning problems.

The main consequence (and motivation) of Theorem~\ref{thm:main} is complexity theoretic. To explain it we briefly recall Khot's
Unique Games Conjecture (UGC), which asserts
 that for every $\e\in
(0,1)$ there exists a prime $p=p(\e)\in \N$ such that no polynomial
time algorithm can perform the following task. The input is a system
of $m$ linear equations in $n$ variables $x_1,\ldots,x_n$, each of
which has the form $x_i-x_j\equiv c_{ij}\mod p$. The algorithm must determine
whether there exists an assignment of an integer value to each
variable $x_i$ such that at least $(1-\e)m$ of the equations are
satisfied, or no assignment of such values can satisfy more
than $\e m$ of the equations. If neither of these possibilities
occurs, then an arbitrary output of the algorithm is allowed. The UGC was introduced by Khot in~\cite{Khot02}, though the above
formulation of it, which is equivalent to the original one, is due
to~\cite{KKMO07}. The use of the UGC as a hardness hypothesis has
become popular over the past decade; we refer to Khot's
survey~\cite{Kho10} for more information on this topic. Saying that the UGC hardness threshold of an optimization problem $\mathscr O$ equals $\alpha\in [1,\infty)$ means that for every $\e\in (0,1)$ there exists a polynomial time algorithm that outputs a number that is guaranteed to be within a multiplicative factor of $\alpha+\e$ from the solution of $\mathscr O$, and that the existence of such an algorithm with approximation guarantee of $\alpha-\e$ would contradict the UGC.

The Kernel Clustering problem is a clustering framework for covariance matrices that originated in the work of Borgwardt, Gretton, Smola and Song~\cite{SSGB07} in the context of Machine Learning. This problem is generic in the sense that it contains well-studied optimization problems as special cases, and its versatility allows one to design a variety of algorithms tailor-made for particular applications (many of these algorithms are at present shown~\cite{SSGB07} to be successful empirically, but not rigorously). The input of the Kernel Clustering problem is an $n\times n$ symmetric positive semidefinite matrix $A=(a_{ij})$ and a $k\times k$ symmetric positive semidefinite matrix $B=(b_{ij})$ called the hypothesis matrix. Think of $n$ as very large and $k$ as small, the goal being to cluster the entries of $A$ into a $k\times k$ matrix that is most correlated with the hypothesis matrix $B$. Formally, given a partition $\{S_1,\ldots,S_k\}$ of $\{1,\ldots,n\}$, form the associated clustered version of $A$ by summing the entries of $A$ over the blocks induced by the partition $\{S_1,\ldots,S_k\}$. One thus obtains a $k\times k$ matrix $C=(c_{ij})$ given by $c_{ij}=\sum_{(s,t)\in S_i\times S_j} a_{st}$. Let $\Clust(A|B)$ denote the maximum of $\sum_{i=1}^k\sum_{j=1}^k c_{ij}b_{ij}$ over all partitions $\{S_1,\ldots,S_k\}$ of $\{1,\ldots,n\}$. We refer to~\cite{SSGB07,KN08,KN10,KN11} for further explanation of this clustering framework, as well as a discussion of important special cases arising from appropriate choices of the hypothesis matrix $B$.

In what follows, an $n\times n$ symmetric positive semidefinite matrix $A=(a_{ij})$ is called centered if $\sum_{i=1}^n\sum_{j=1}^n a_{ij}=0$, and it is called spherical if $a_{ii}=1$ for all $i\in \{1,\ldots,n\}$.

\begin{theorem}[Main theorem; complexity theoretic formulation]\label{thm:kernel}
Let $\mathscr O$ be the following optimization problem. The input is an $n\times n$ symmetric positive semidefinite centered matrix $A=(a_{ij})$, and also a $4\times 4$ symmetric centered and spherical positive semidefinite matrix $B=(b_{ij})$. The goal is to compute in polynomial time the quantity $\Clust(A|B)$. Then the UGC hardness threshold of $\mathscr O$ equals $\frac{2\pi}{3}$.
\end{theorem}

A polynomial time algorithm  with approximation ratio $\frac{2\pi}{3}+o(1)$ for the problem $\mathscr O$ of Theorem~\ref{thm:kernel} was designed in~\cite{KN08}. The fact that Theorem~\ref{thm:main} implies the matching UGC hardness result was also proved in~\cite{KN08}. More generally, it was shown in~\cite{KN08} that the validity of the $m$-dimensional Propeller Conjecture for some $m\ge 3$ would imply that the UGC hardness threshold of the variant of $\mathscr O$ with $B$ being an $m\times m$ matrix equals $\frac{8\pi}{9}\left(1-\frac{1}{m}\right)$.

This is not the first time
that extremal problems in the measure space $(\R^m,\gamma_m)$ arose
from investigations in complexity theory; see for example the
Majority is Stablest Conjecture of Khot, Kindler, Mossel and
O'Donnell~\cite{KKMO07}, which was solved by Mossel, O'Donnell, and
Oleszkiewicz~\cite{MOO} via  a reduction to a classical
isoperimetric inequality of Borell~\cite{Bor85}. A special feature
of the Propeller Conjecture is that its conclusion is quite
surprising: despite allowing for partitions of $\R^m$ into $m+1$
sets, the optimal partition has only three nonempty sets. Note that
this degeneracy property occurs for the first time when $m=3$, since
in the result of~\cite{KN08} for partitions of $\R^2$ into three
sets, all three sets are indeed present in the optimal partition.
Thus Theorem~\ref{thm:main} is a proof of the first nonintuitive
case of the Propeller Conjecture: a verification of a prediction
about a clean yet unexpected geometric phenomenon that arose from an
investigation in computational complexity. To the best of our
knowledge this is the first time that such a development occurred in
this theory.

Our proof of Theorem~\ref{thm:main} proceeds as follows. First, using reductions of the Propeller Conjecture that were obtained in~\cite{KN08}, we show that it suffices to prove an analogous partitioning problem for the sphere $S^2$. Namely,  given a partition of $S^2$ into four spherical triangles, the goal is to show that the sum of the squares of lengths of the moments (with respect to the surface area measure on $S^2$) of these spherical triangles cannot exceed $9\pi^2/4$. Using additional geometric arguments, we further reduce the question to an optimization problem over a three dimensional search space. We thus obtain a coordinate system with three degrees of freedom, with respect to which one must solve a rather complicated nonlinear optimization problem. After proving several additional estimates that serve to further reduce the search space and give crucial modulus of continuity estimates, we show that it suffices to check that the desired estimate holds true for a finite list of spherical partitions (arising from an appropriate net of the search space). This list of inequalities that must be proved is explicit but very large, so we proceed to check it in a fully rigorous computer-assisted fashion. That is, our computation carefully accounts for all the rounding errors; equivalently the computation can be viewed as an implementation in our setting of interval arithmetic (see~\cite{Han92}).

The role and implications of rigorous computer-assisted proofs has been discussed at length in the literature; many such discussions appear in papers that establish striking results via a proof that has a computer-assisted component (famous examples include~\cite{AH77,AHK77,HS00,GMT03,Hal05,GMM09}). We see no reason to include here a new treatment of this topic. Zwick's discussion in~\cite{Zwi02} is an excellent reference for a well thought out explanation of the role of computer-assisted proofs in mathematics and computer science\footnote{Zwick's paper~\cite{Zwi02} is another instance of a computer-assisted proof in the context of approximation algorithms. While the present article deals with an integrality gap lower bound for a semindefinite program (the left hand side of~\eqref{eq:our gro}), Zwick proves an integrality gap upper bound, i.e., he uses a rigorous computer-assisted argument to show that a certain algorithm performs well rather than to prove a hardness result.}. The essence of the argument can be best conveyed by  quoting Zwick directly~\cite[Sec.~7]{Zwi02}:
\begin{quote}
{
``A typical computer assisted proof, like ours, is essentially composed of two steps. The first step says something like: ``Here is program P. If program P, when executed by an abstract computer, outputs YES, then the theorem is true, because . . . ''. The second step says something like: ``I ran program P and it said YES!'', or more specifically: ``I compiled program P using compiler A, under operating system B, and ran it on processor C, again under operating system B. It said YES!''.''}
\end{quote}
Zwick then proceeds to explain that
\begin{quote}
``The first step is a conventional mathematical proof. It simply argues that a certain program has certain properties. Such arguments are common in computer science. Anyone who objects to this step of the proof should find a flaw, or a gap, in the arguments made.''
\end{quote}
Nevertheless, Zwick explains that the second step, i.e., the claim that the compilation of the program gave the desired result, is not of the same nature:
\begin{quote}
``The second step is more problematic. It is certainly not a proof in the conventional mathematical sense of the word. Many things can go wrong here. Program P may not have been compiled correctly by compiler A. Or, due to some bug in operating system B, the program did not run as intended. Or, processor C may suffer from a design flaw that causes an incorrect execution of the program, or perhaps the specific processor used has a short-circuit somewhere, and so on. It is possible to reduce the likelihood of such problems by
compiling the program using several different compilers, and running it on different processors. But, while we may eventually be able to produce mathematical correctness proofs for compilers, operating systems, and the hardware design of the processors, it seems that we would never be able to produce a mathematical proof that no hardware fault occurred during a specific computation.

 What we can do, is instruct the computer executing program P to print a trace of the execution. This trace is a mathematical proof, though perhaps a not so inspiring one. Like any mathematical proof it should be checked carefully to make sure that it is correct. In many cases, however, this is not humanly possible. The main problem with computer assisted proofs, therefore, is that they are usually too long. They are also, in most cases, less insightful than conventional proofs."
\end{quote}

The bulk of the work presented in this article consists of conventional geometric and analytic proofs, resulting in a reduction of the problem to a region where the desired inequality holds with ``room to spare''. This extra room allows us to complete the proof by checking a finite list of concrete inequalities (that we write explicitly) on a  certain finite net. Each such inequality can be checked by hand, but the number of such checks is large, and it would be unrealistic (and probably unilluminating) to complete this final check manually. The code that we used is publicly available at the following link, which contains a command-line user interface so that interacting with the code is easy.
\begin{center}
\url{http://cims.nyu.edu/~aukosh/propeller/propeller.html}
\end{center}
The code for this proof was run many times on several different computers. The technical requirements are quite modest, with the processor speed only affecting the total run time, and the total memory requirement being less than 200 Mb. One can run the code on anything from a dedicated compute server to a laptop computer (we did both).
We ran the code in two different ways. The simplest way was to run it in a serial fashion. That is, we ran the program as a single instance which performed each step of the algorithm described in this paper one after the other. This took 37 hours to complete on a single processor of a compute server with 3.3 GhZ Intel Xenon processors. The second way we ran it was to parallelize the algorithm by breaking up the domain of interest into seven distinct portions and running separate instances on each. This procedure took eight hours when split between seven 3.3 GhZ Intel Xenon processors. (We also ran the program split between seven 2.4 GhZ AMD Opteron processors; this procedure took twelve hours.)

Needless to say, our proof of the Propeller Conjecture in $\R^3$ leaves something to be desired, since we do not have a short explanation of the validity of the computations in the final step of our proof. It is conceptually important to have a proof of this result, despite the fact that it ends in a lengthy computation,  not only because it yields interesting results in mathematics and computational complexity. Importantly, since the Propeller Conjecture is a nonintuitive prediction that arose from investigations in computational complexity, knowing that the first nonintuitive case of this conjecture is indeed true puts the full Propeller Conjecture, and consequently the link between geometry and algorithms that it describes, on a stronger footing. Thus, in addition to yielding a remarkably involved UGC hardness result, our theorem will hopefully invigorate future research on this problem that might lead to a traditional mathematical proof  of the Propeller Conjecture in $\R^3$ that can hopefully be extended to higher dimensions.

A key feature of the Propeller Conjecture is that it asserts that a natural optimization problem exhibits intermediate-dimensional symmetry breaking. There are precedents of results of this type that have been proved in the literature, mostly using Fourier-analytic methods. For this reason we are hopeful that a clean and shorter proof of the Propeller Conjecture will eventually be found. Consider for example the problem of sharp Khinchine inequalities: if $\e_1,\ldots,\e_n$ are i.i.d. symmetric random variables taking values in $\{-1,1\}$ and $p\in [1,2)$, the goal is to compute the minimum of $\E\left[\left|\sum_{i=1}^n a_i\e_i\right|^p\right]$ over all unit vectors $a=(a_1,\ldots,a_n)\in \R^n$. For $p=1$ it was conjectured by Littlewood (see~\cite{Hal75}) that this minimum occurs at $a=(1,1,0,\ldots,0)/\sqrt 2$. Littlewood's conjecture was solved affirmatively by Szarek~\cite{Sza76} (see also~\cite{Tom87,LO94}). Haagerup~\cite{Haa81} proved that the same two dimensional symmetry breaking occurs for $p\in [1,p_0]$, where $p_0=1.87...$ is the solution of the equation $2\Gamma((p+1)/2)=\sqrt{\pi}$, i.e., the unit vector that minimizes $\E\left[\left|\sum_{i=1}^n a_i\e_i\right|^p\right]$ equals $(1,1,0,\ldots,0)/\sqrt 2$ for $p\in [1,p_0]$ and for $p\in [p_0,2]$ this minimum occurs at $a=(1,\ldots,1)/\sqrt n$. Another famous result of this type is Ball's cube slicing theorem~\cite{Bal86} (resolving a conjecture of Hensley~\cite{Hen79}), asserting that the hyperplane section of $[-1,1]^n$ with maximal $(n-1)$-dimensional volume is perpendicular to $(1,1,0,\ldots,0)$; see~\cite{OP00} for the corresponding result for complex scalars, as well as~\cite{Bal95,BN02} for an analogous statement for projections. We refer to~\cite{NP00} for a unified treatment of the results of Szarek, Haagerup, and Ball. There is also a conjecture of  Milman predicting a symmetry breaking phenomenon for extremal volumes of slabs in the cube $[-1,1]^n$; see~\cite{BK03,KK11} for partial results along these lines. We mentioned the above statements since we believe that among the literature they are most similar to the Propeller Conjecture, and perhaps (with much more work) related methods could be be used to address the Propeller Conjecture as well.

This article is organized as follows. In Section~\ref{sec:overview} we present the reduction of the Propeller Conjecture to a  certain optimization problem for spherical partitions, and we explain the main ingredients of our proof, including the conventional geometric/analytic arguments, as well as the computer-assisted component. The proofs of the geometric and analytic results leading to the final computational step are contained in Section~\ref{sec:identities} and Section~\ref{secnet}. Section~\ref{sec:implement} contains a detailed explanation of how the numerical step is implemented so as to account for all possible rounding errors. We remark that our arguments extend mutatis mutandis to higher dimensions. For the sake of simplicity we present the entire argument in $\R^3$, since while it is conceivable that the scheme presented here can yield a computer-assisted proof of the Propeller Conjecture in higher dimensions, at some fixed dimension the computer-assisted component of the proof will become unfeasible. Now that we know that a nonintuitive case of the Propeller Conjecture is indeed correct, the next natural step is to search for a proof that extends to all dimensions, rather than attempting to prove a few more cases in low dimensions.

\section{An overview of the proof of Theorem~\ref{thm:main}}\label{sec:overview}

From now on assume for the sake of eventually obtaining a contradiction that $\{A_i\}_{i=1}^4$ is a partition of $\R^3$ into measurable sets that violates the propeller conjecture, i.e.,
\begin{equation}\label{eq:contradiction}
\sum_{i=1}^4 \left\|\int_{A_i}xd\gamma_3(x)\right\|_2^2>
\frac{9}{8\pi}.
\end{equation}
Assume moreover that the maximum of $\sum_{i=1}^4 \left\|\int_{B_i}xd\gamma_3(x)\right\|_2^2$ over all measurable partitions $\{B_i\}_{i=1}^4$ of $\R^3$ is attained at $\{A_i\}_{i=1}^4$. For a proof that this maximum is indeed attained, see~\cite[Lem.~3,1]{KN08}. Using this maximality, it follows from~\cite[Lem.~3.3]{KN08} that (up to measure zero corrections) the $A_i$ are cones with cusp at the origin, and if we write  $\zeta_i=\int_{A_i} xd\gamma_3(x)$ then $\zeta_1+\zeta_2+\zeta_3+\zeta_4=0$ and for each $i\in \{1,2,3,4\}$ we have
\begin{equation}\label{eq:zeta}
A_i=\left\{x\in \R^3:\ \langle x,\zeta_i\rangle=\max_{j\in \{1,2,3,4\}} \langle x,\zeta_j\rangle\right\}.
\end{equation}
By~\cite[Lem.~3.3 \& Cor.~3.4]{KN08} the vectors $\{\zeta_i\}_{i=1}^4$ are distinct, nonzero and not coplanar.

Since the sets in question are cones, it is beneficial to study them in terms of their intersection with the sphere $S^2$, i.e., define
\begin{equation}\label{eq:def T_i}
T_i=A_i\cap S^2,
\end{equation}
and letting $\sigma$ denote the surface area measure on $S^2$ (thus $\sigma(S^2)=4\pi$), define
\begin{equation}\label{eq:def z_i}
z_i=\int_{T_i} xd\sigma(x)=\sqrt{2\pi^3}\cdot\zeta_i,
\end{equation}
where the last equality in~\eqref{eq:def z_i} follows from integration in polar coordinates. Being a constant multiple of  the vectors $\{\zeta_i\}_{i=1}^4$, the vectors $\{z_i\}_{i=1}^4$ are also distinct, nonzero, not coplanar, and satisfy $z_1+z_2+z_3+z_4=0$. Moreover,
\begin{equation}\label{ONE1}
T_{i}\stackrel{\eqref{eq:zeta}\wedge\eqref{eq:def T_i}\wedge\eqref{eq:def z_i}}{=}\left\{x\in S^{2}\colon \langle x,z_{i}\rangle=\max_{j\in\{1,2,3,4\}}\langle x,z_{j}\rangle\right\}.
\end{equation}
Consequently,
\begin{equation}\label{eq:on edge}
\forall i,j\in \{1,2,3,4\},\quad x\in T_{i}\cap T_{j}\implies \langle z_{i},x\rangle=\langle z_{j},x\rangle.
\end{equation}

Fix $\ell\in \{1,2,3,4\}$. Since $\{z_i\}_{i=1}^4$ are not coplanar there exists a unique $v_\ell\in S^2$ satisfying
\begin{equation}\label{eq:def v}
\forall i,j\in \{1,2,3,4\}\setminus \{\ell\},\quad\langle z_i,v_\ell\rangle=\langle z_j,v_\ell\rangle>\langle z_\ell,v_\ell\rangle .
\end{equation}
The vectors $\{v_\ell\}_{\ell=1}^4$ are the vertices of the partition $P=\{T_1,T_2,T_3,T_4\}$ of $S^2$. A simple argument presented in Section~\ref{sec:identities} shows that $v_i\notin \{-v_j,v_j\}$ if $i\neq j$ and each spherical triangle $T_i$ is contained in an open hemisphere of $S^2$. Moreover, for distinct $i,j,\ell\in \{1,2,3,4\}$ we have $\det(v_i,v_j,v_\ell)\neq 0$, and if we define
\begin{equation}\label{eq:def thetas}
\theta_{ij} \colonequals \arccos(\langle v_{i},v_{j}\rangle)\qquad \mathrm{and}\qquad
\Theta_{ij\ell} \colonequals
\arccos\left(\left\langle\frac{(v_{i}\times v_{j})}{\vnormf{v_{i}\times v_{j}}_{2}},
\frac{(v_{j}\times v_{\ell})}{\vnormf{v_{j}\times v_{\ell}}_{2}}\right\rangle\right),
\end{equation}
then it is also argued in Section~\ref{sec:identities} that $\theta_{ij},\Theta_{ij\ell}\in (0,\pi)$.  Observe $\theta_{ij}=\theta_{ji}$, $\Theta_{ijk}=\Theta_{kji}$, and $\Theta_{ij\ell}$ is the spherical angle, at vertex $v_{j}$, of the spherical triangle with vertices $\{v_{i},v_{j},v_{\ell}\}$; the cosine of this angle is exactly the inner product of unit normals of the two planes containing $\{v_{i},v_{j},0\}$ and $\{v_{j},v_{\ell},0\}$ respectively. See Figure~\ref{fig1} for a schematic description of the notation.
\begin{figure}[htbp!]
\subfigure[]{
   \centering
   \def\svgwidth{\figoneawidth} 
   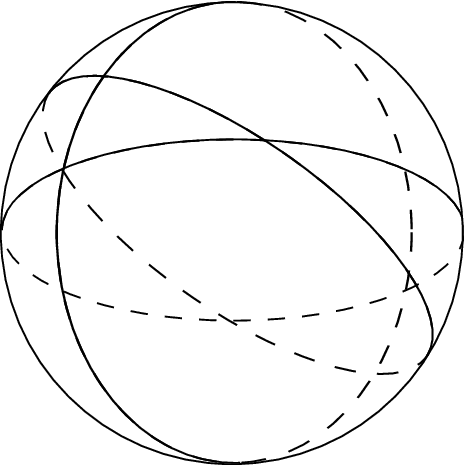
}
\subfigure[]{
   \centering
   \def\svgwidth{\figonebwidth} 
   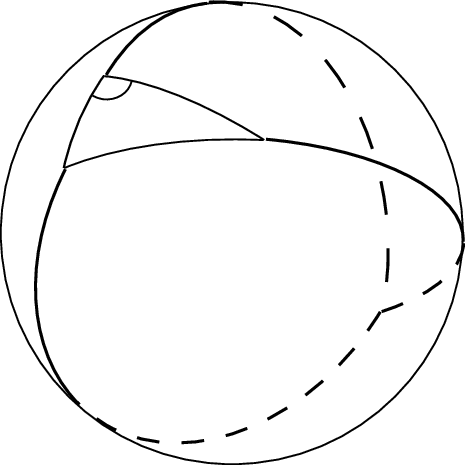
}
\caption{(a) Spherical triangle constructed with three great circles (b) Partition of the sphere into four spherical triangles.}
\label{fig1}
\end{figure}

In order to proceed we need to have a formula for the spherical moment $z_4$ in terms of the vertices $\{v_1,v_2,v_3\}$ of the spherical triangle $T_4$. In order to do so, assume that $\det(v_1,v_2,v_3)>0$. Then
\begin{flalign}
z_{4}
&=\frac{1}{2}\left(\theta_{12}\frac{v_{1}\times v_{2}}{\vnorm{v_{1}\times v_{2}}_{2}}
+\theta_{23}\frac{v_{2}\times v_{3}}{\vnorm{v_{2}\times v_{3}}_{2}}
+\theta_{31}\frac{v_{3}\times v_{1}}{\vnorm{v_{3}\times v_{1}}_{2}}\right)\label{onec0}.
\end{flalign}
Identity~\eqref{onec0} was proved by Minchin in 1877; see~\cite[p.~259]{minchin77}. We also found identity~\eqref{onec0} in~\cite{brock74}, which is  a military publication that is not publicly available (it isn't classified: we purchased access to it). Since both references for~\eqref{onec0} are hard to find, we include a brief derivation of it in Proposition~\ref{prop1.5}, using the Cauchy projection formula (see~\cite[p.~25]{Kol05}).

Using~\eqref{onec0} and geometric arguments, we show in Lemma~\ref{prop3} that the vertex $v_4$ is determined by the vertices $\{v_1,v_2,v_3\}$ as follows; an analogous formula expresses each vertex of the partition in terms of the other three vertices.
\begin{equation*}
v_{4}=-\frac{1}{\sqrt{\lambda}}\left(\frac{\sin\theta_{23}}{\theta_{23}}v_{1}+\frac{\sin\theta_{13}}{\theta_{13}}v_{2}+
\frac{\sin\theta_{12}}{\theta_{12}}v_{3}\right),
\end{equation*}
where
\begin{multline}\label{one1.60}
\lambda
\colonequals\vnorm{\frac{\sin\theta_{23}}{\theta_{23}}v_{1}+\frac{\sin\theta_{13}}{\theta_{13}}v_{2}+\frac{\sin\theta_{12}}{\theta_{12}}v_{3}}_{2}^{2}
 = \frac{\sin^{2}\theta_{23}}{\theta_{23}^{2}}+\frac{\sin^{2}\theta_{13}}{\theta_{13}^{2}}+\frac{\sin^{2}\theta_{12}}{\theta_{12}^{2}}\\
+2\frac{\sin\theta_{23}\sin\theta_{13}}{\theta_{23}\theta_{13}}\cos\theta_{12}
+2\frac{\sin\theta_{23}\sin\theta_{12}}{\theta_{23}\theta_{12}}\cos\theta_{13}
+2\frac{\sin\theta_{12}\sin\theta_{13}}{\theta_{12}\theta_{13}}\cos\theta_{23}.
\end{multline}
In Proposition~\ref{prop6} we derive the following restriction on products of opposite edges.
\begin{equation*}\label{one20}
\frac{\sin\theta_{12}\sin\theta_{34}}{\theta_{12}\theta_{34}}
=\frac{\sin\theta_{13}\sin\theta_{24}}{\theta_{13}\theta_{24}}
=\frac{\sin\theta_{14}\sin\theta_{23}}{\theta_{14}\theta_{23}}.
\end{equation*}
Also, in Lemma~\ref{lem:squares} we obtain the following expression for the objective function at the partition $P=\{T_{i}\}_{i=1}^{4}$ as the sum of the squares of the angles between the vertices.
\begin{equation}\label{eq:sum of squares}
F(P)\colonequals\|z_1\|_2^2+\|z_2\|_2^2+\|z_3\|_2^2+\|z_4\|_2^2=\theta_{12}^2+\theta_{13}^2+\theta_{14}^2+\theta_{23}^2+\theta_{24}^2+\theta_{34}^2.
\end{equation}
Define
\begin{equation}\label{one3.00}
\gamma\colonequals\frac{\sin\theta_{23}}{\theta_{23}}\cos\theta_{12}+\frac{\sin\theta_{13}}{\theta_{13}}+\frac{\sin\theta_{12}}
{\theta_{12}}\cos\theta_{23}.
\end{equation}
In Lemma~\ref{prop7} we obtain the following useful restriction on a single spherical triangle appearing in the partition $\{T_1,T_2,T_3,T_4\}$.
\begin{equation}\label{eq:gamma identity}
\cos\left(\sqrt{\lambda-\gamma^{2}}\frac{\theta_{23}\theta_{12}}{\sin\theta_{23}\sin\theta_{12}}\right)
=\cos\left(\sqrt{\theta_{12}^{2}+\theta_{23}^{2}+2\theta_{12}\theta_{23}\cos\Theta_{123}}\right)=-\frac{\gamma}{\sqrt{\lambda}},
\end{equation}
where $\lambda$ is given in~\eqref{one1.60}.

Observe that~\eqref{eq:gamma identity} is a restriction involving data from only one triangle $T_{4}$, as we see by the definitions of $\gamma$ and $\lambda$.  Moreover, by symmetry, every cyclic permutation of $\{\theta_{12},\theta_{23},\theta_{13}\}$ must also satisfy~\eqref{eq:gamma identity}. And, every triangle of $P$ must satisfy~\eqref{eq:gamma identity}, with the indices of the $\theta_{ij}$ substituted appropriately. To understand the ramifications of~\eqref{eq:gamma identity} assume that $\theta_{12}=\theta_{23}=\theta_{13}\equalscolon\theta$. Then \eqref{eq:gamma identity} simplifies to
\begin{equation}\label{one5.1}
\cos\left(\sqrt{2(2\cos\theta +1)(1-\cos\theta)}\frac{\theta}{\sin\theta}\right)=-\frac{1}{\sqrt{3}}\sqrt{2\cos\theta +1}.
\end{equation}
The only $\theta\in(0,2\pi/3)$ satisfying this equation are $\theta=\arccos(-1/3)\approx 1.9106332362490187$, and $\theta\approx 1.5379684120790425$.  The first solution corresponds to the partition whose vertices are those of a regular simplex
inscribed in the sphere, and the second intuitively corresponds to a critical point ``between'' the regular partition and the propeller partition (note that $\theta=2\pi/3$ implies $\det(v_{1},v_{2},v_{3})=0$, and $\theta>2\pi/3$ cannot be achieved by a spherical triangle).  The above special values of $\theta$ will play an important role in the ensuing computations.

As explained in the discussion following Lemma~\ref{lem:squares}, a combination of~\eqref{eq:sum of squares} and~\eqref{eq:gamma identity} can be used to express the value of our objective function at the maximizing partition $P$ in terms of the data from a single spherical triangle in $P$. Specifically, we have
\begin{equation}\label{eq:F single tri}
F(P)= 3\left(\theta_{12}^2+\theta_{23}^2+\theta_{13}^2\right)+2\cos(\Theta_{213})\theta_{12}\theta_{13}+2\cos(\Theta_{123})\theta_{12}\theta_{23}
+2\cos(\Theta_{231})\theta_{23}\theta_{13}.
\end{equation}
In particular, it follows that the right hand side of~\eqref{eq:F single tri} must be the same for each set of data for all four triangles of $P$.

Up to now we only discussed  identities that the extremal partition $P$ must satisfy. In order to proceed we need to prove some a priori estimates on the various parameters in question. First, let $M=\max\{\theta_{12},\theta_{13},\theta_{23}\}$. It follows from~\eqref{eq:F single tri}  that $F(P)\le 14 M^2$ (using the fact that at least one $\Theta_{ijk}$ must be greater than $\frac{\pi}{3}$). Our contrapositive assumption~\eqref{eq:contradiction} combined with~\eqref{eq:def z_i} implies that $F(P)>\frac{9\pi^2}{4}$, so we deduce that all the triangles in $P$ must have an edge of length greater than $\frac{3\pi}{2\sqrt{14}}>\frac54$. Additional arguments (that are significantly more involved technically) imply that for all distinct $i,j,\ell\in \{1,2,3,4\}$ we have $\theta_{ij}\le \pi-\frac12$ (Lemma~\ref{lemma4}) and $\theta_{ij},\sin\Theta_{ij\ell}>\frac1{10}$ (Lemma~\ref{thm1}). In Lemma~\ref{lemma7} we prove that $(\theta_{12},\theta_{13},\theta_{23})$ must be outside the $\ell_2$ ball of radius $\frac{1}{100}$ centered at $(\theta,\theta,\theta)$ for $\theta=\arccos\left(-\frac{1}{3}\right)$ and $\theta=1.53796841207904$. This excludes from our search space  balls centered at the two solutions of~\eqref{one5.1} that were discussed above. In Lemma~\ref{rk5} we show that, up to a relabeling of the spherical triangles $T_1,T_2,T_3,T_4$, we may assume that $\sqrt{\lambda}>\frac{9}{50}$.

It turns out that the above estimates suffice in order to conclude the proof of Theorem~\ref{thm:main} via a search over a sufficiently fine net. To explain this endgame, note that by the spherical law of cosines we have $\cos \theta_{13}=\cos\theta_{12}\cos\theta_{23}+\sin\theta_{12}\sin\theta_{23}\cos\Theta_{123}$. It therefore follows from~\eqref{one1.60}, \eqref{eq:gamma identity}, \eqref{one3.00} that if we define  $h,\lambda,\gamma:[0,\pi]^3\to \R$ by
$$
\lambda(x,y,z)\colonequals\frac{\sin^2 x}{x}+\frac{\sin^2 y}{y}+\frac{\sin^2 z}{z}+2\frac{\sin x\sin y}{xy}\cos z+2\frac{\sin x\sin z}{xz}\cos y+2\frac{\sin y\sin z}{yz}\cos x,
$$
$$
\gamma(x,y,z)\colonequals\frac{\sin z}{z}\cos x+\frac{\sin y}{y}+\frac{\sin x}{x}\cos z,
$$
$$
h(x,y,z)\colonequals \sqrt{\lambda(x,y,z)}\cos\left(\sqrt{x^2+z^2+2x
z\frac{\cos y-\cos x\cos z}{\sin x\sin z}}\right)+\gamma(x,y,z),
$$
then $h(\theta_{12},\theta_{13},\theta_{23})=0$. This identity must hold  for all cyclic permutations of the indices $\{1,2,3\}$, so we also have the two identities $h(\theta_{13},\theta_{23},\theta_{12})=h(\theta_{23},\theta_{12},\theta_{13})=0$. It follows that if we define $H=(H_1,H_2,H_3):[0,\pi]^3\to \R^3$ by $H(x,y,z)=(h(y,z,x),h(x,y,z),h(z,x,y))$ then at the extremal partition $P$ we have $H(\theta_{12},\theta_{13},\theta_{23})=0$.

Our strategy is therefore as follows. At every point $q$ in the search space given by the constraints described above (specifically, see the system of equations~\eqref{three1.1}, excluding also the two balls described in Lemma~\ref{lemma7}), we will show that either $H(q)\neq 0$ or $F_0(q)<9\pi^2/4$, where $F_0$ is given by the right hand side of~\eqref{eq:F single tri} (with the understanding that one expresses $\cos\Theta_{123}, \cos\Theta_{213}, \cos\Theta_{231}$ in terms of $\theta_{12},\theta_{23},\theta_{13}$ using the spherical law of cosines, so that $F_0$ is a function of $\theta_{12},\theta_{23},\theta_{13}$). Due to the above discussion, such an assertion will complete the proof of Theorem~\ref{thm:main}. Moreover, it turns out that this assertion holds ``with room to spare'', making a computer-assisted verification feasible.

To this end, we need to have estimates on the modulus of continuity of $H$ and $F_0$. Such estimates are complicated but can be proved using elementary considerations; see Lemma~\ref{lemma6} and the discussion immediately following it. Consequently, if we are given a point $q$ in our restricted search space and $\tau\in (0,1)$  for which either $|H(q)|>\tau$  or $F_0(q)<9\pi^2/4-\tau$, then it would follow that for some $r>0$ that depends on $\tau$ via our modulus of continuity estimates, in the entire ball of radius $r$ centered at $q$ either $H\neq 0$ or $F_0<9\pi^2/4$.

From this procedure we achieve a process  by which we can iteratively remove from our search space macroscopically large regions of controlled size in which we are guaranteed that a counterexample to the Propeller Conjecture cannot exist. Our program proceeds to remove such balls until it eventually exhausts the entire search space, thus arriving at the conclusion that there is no counterexample to the Propeller Conjecture.

In order to make such a procedure rigorous, we  carefully account for all possible rounding errors. We do this quite conservatively, i.e., while significantly overestimating the magnitude of all possible errors, as explained in detail in Section~\ref{sec:implement}. Note that our  computer-assisted proof does not compute any computationally complicated expressions such as, say, definite integrals or implicitly defined functions: it only checks inequalities between expressions involving elementary combinations of trigonometric functions and square roots.


In summary, the validity of the Propeller Conjecture in $\R^3$ has mainly conceptual implications, and what remains open seems to be of a more technical nature: to perhaps find a clever transformation, e.g., as in Ball's cube slicing theorem~\cite{Bal86}, that allows one to prove the theorem analytically rather than resorting to a transversal of a net in a region where the desired inequality is actually quite weak. At present, we have a clean geometric argument which addresses directly the region where the Propeller Conjecture is most subtle, and in the remainder of the search space we do something rather crude. This has two conceptual consequences: the Propeller Conjecture as a nonintuitive link between complexity theory and geometry is correct, and moreover one has a geometric/analytic proof of this conjecture in a region where it is tightest. Nevertheless, we are lacking a technical idea that allows us to address the remaining region in a way that does not resort to ``brute force''. It remains a challenge to find such an idea, with the hope that it will pave the way to a proof of the Propeller Conjecture in all dimensions.





\section{Proofs of basic identities and estimates}\label{sec:identities}

This section contains proofs the geometric identities and inequalities that were stated in Section~\ref{sec:overview}. Recall that we are assuming that $P=\{T_1,T_2,T_3,T_4\}$ is a partition of $S^2$ into four spherical triangles, and that $\{z_i\}_{i=1}^4$, as defined in~\eqref{eq:def z_i}, are the corresponding spherical moments. We are also making the contrapositive assumption that our objective function
\begin{equation}\label{eq:the F}
F(P)=\|z_1\|_2^2+\|z_2\|_2^2+\|z_3\|_2^2+\|z_4\|_2^2
\end{equation}
(recall~\eqref{eq:sum of squares}) exceeds $9\pi^2/4$, which is the maximal value of $F(\cdot)$ that is predicted by the Propeller Conjecture. We also assume that $P$ maximizes $F$. The vertices of the partition $\{v_i\}_{i=1}^4$ were defined in Section~\ref{sec:overview}, along with the angles $\theta_{ij}$ and $\Theta_{ij\ell}$ as given in~\eqref{eq:def thetas}.

\begin{prop}\label{prop1}\textup{(a)} $z_1+z_2+z_3+z_4=0$. 
\begin{itemize}
\item[(b)] Suppose two triangles $T_{i},T_{j}\in P$ share an edge $E\colonequals T_{i}\cap T_{j}$.  Then for all $e\in E$ we have $\langle z_i,e\rangle=\langle z_j,e\rangle$.  In particular, if $\ell\notin\{i,j\}$  then $\langle z_{i},v_{\ell}\rangle=\langle z_{j},v_{\ell}\rangle$.
\item[(c)] For all distinct $i,j\in \{1,2,3,4\}$ we have $\langle z_{i},v_{i}\rangle=\langle -3z_{j},v_{i}\rangle$.
\item[(d)] For all distinct $i,j\in \{1,2,3,4\}$ we have $v_j\notin \{-v_i,v_i\}$.
\item[(e)] For all distinct $i,j,\ell\in\{1,2,3,4\}$ we have $\det(v_{i},v_{j},v_{\ell})\neq0$ , each $T_{i}$ is contained in an open hemisphere, and $0<\theta_{ij},\Theta_{ij\ell}<\pi$.
\end{itemize}
\end{prop}
\begin{proof}
Parts (a) and (b) were already proved in Section~\ref{sec:overview}; for part (b) see~\eqref{eq:on edge}. A combination of~\eqref{eq:def v} and part (a) implies part (c). For $\ell\in \{1,2,3,4\}$ let $\Pi_{\ell}$ be the unique plane containing $\{z_{i}\}_{i\in \{1,2,3,4\}\setminus \{\ell\}}$.  A vector is perpendicular to $\Pi_{\ell}$ if and only if it has equal inner product on $\{z_{i}\}_{i\in \{1,2,3,4\}\setminus \{\ell\}}$, so $\R v_{\ell}$ is also perpendicular to $\Pi_{\ell}$.  Since $\{z_{i}\}_{i=1}^{4}$ are not coplanar the planes $\{\Pi_{\ell}\}_{\ell=1}^{4}$ have distinct perpendiculars.  So $v_{1},v_2,v_3,v_4$ are distinct with $v_{i}\notin \R v_{j}$ for $i\neq j$.  This establishes (d). For every $\ell\in \{1,2,3,4\}$ consider the set
$$
U_\ell\colonequals \left\{x\in S^2:\ \min_{i\in \{1,2,3,4\}\setminus \ell} \langle x,z_i-z_\ell\rangle>\frac12 \langle v_\ell,z_i-z_\ell\rangle\right\}.
$$
It follows from~\eqref{eq:def v} that $U_\ell$ is an open neighborhood of $v_\ell$, and in combination with~\eqref{ONE1} we know that
\begin{equation}\label{eq:neighborhood intersection}
i\in \{1,2,3,4\}\setminus\{\ell\}\implies U_\ell\cap T_i=\left\{x\in U_\ell\cap S^2:\ \langle x,z_i\rangle=\max_{j\in \{1,2,3,4\}\setminus\{\ell\}} \langle x,z_j\rangle\right\}.
\end{equation}
Write $I=\{1,2,3,4\}\setminus\{\ell\}$. For every distinct $i,k\in \{1,2,3,4\}$ the set $\{x\in S^2:\ \langle x,z_i\rangle \ge \langle x,z_j\rangle\}$ is a half sphere, and therefore the set  $T_{i,\ell}\colonequals \{x\in S^2:\ \langle x,z_i\rangle =\max_{j\in I} \langle x,z_j\rangle\}$ is strictly contained in a half sphere (since the $z_j$ are distinct and nonzero). Note that $\partial T_{i,\ell}$ is the union of two half circles that meet at antipodal points, and $v_\ell$ is one of these antipodal points. So, the spherical angle of $T_{i,\ell}$ at the vertex $v_\ell$ is in $(0,\pi)$. By~\eqref{ONE1} and~\eqref{eq:neighborhood intersection} this spherical angle is identical for $T_{i,\ell}$ and $T_i$. This shows that $\Theta_{i\ell j}\in (0,\pi)$ for all distinct $i,j,\ell\in \{1,2,3,4\}$. Also $\theta_{ij}\in (0,\pi)$ since $v_j\notin \{-v_i,v_i\}$. Next, assume for the sake of contradiction that $\det(v_i,v_j,v_\ell)=0$. Since $v_i\notin \R v_j$ for distinct $i,j\in \{1,2,3,4\}$, it follows that
$v_i=\alpha_j v_j+\alpha_\ell v_\ell$ with $\alpha_j,\alpha_\ell\in \R\setminus\{0\}$. Hence $v_i\times v_j=\alpha_\ell v_\ell\times v_j$ with $\alpha_\ell\neq 0$, contradicting $\Theta_{ij\ell}\in (0,\pi)$. We have already seen that $T_{i,\ell}$ is contained a closed hemisphere, and this containment
  can be chosen such that $T_{i,\ell}$ intersected with the boundary of the hemisphere is $\{v_{\ell},-v_{\ell}\}$. But from~\eqref{ONE1} we know that $T_i\subset T_{i,\ell}$ and $T_i$ avoids a neighborhood of $\{-v_{\ell}\}$, so each $T_{i}$ is contained in an open hemisphere. This concludes the proof of (e).
\end{proof}


We now give an explicit formula for the spherical moments $\{z_{i}\}_{i=1}^4$; as explained in Section~\ref{sec:overview}, this formula can be found in~\cite[p.~259]{minchin77} and~\cite[p.~6]{brock74}.  Since these references are hard to find,  we provide a proof.

\begin{prop}\label{prop1.5}
Suppose that the spherical triangle $T_{4}\subset S^{2}$ has vertices $\{v_{1},v_{2},v_{3}\}$ satisfying $\det(v_{1},v_{2},v_{3})>0$.  Then the vector $z_{4}=\int_{T_{4}}xd\sigma\in\R^{3}$ satisfies
\begin{flalign}
z_{4}
&=\frac{1}{2}\left(\theta_{12}\frac{v_{1}\times v_{2}}{\vnorm{v_{1}\times v_{2}}_{2}}
+\theta_{23}\frac{v_{2}\times v_{3}}{\vnorm{v_{2}\times v_{3}}_{2}}
+\theta_{31}\frac{v_{3}\times v_{1}}{\vnorm{v_{3}\times v_{1}}_{2}}\right)\label{onec}.
\end{flalign}
Consequently
\begin{equation}\label{oneb}
\langle z_{4},v_{1}\rangle
= \left\langle\frac{\theta_{23}}{2}\frac{v_{2}\times v_{3}}{\vnorm{v_{2}\times v_{3}}_{2}},v_{1}\right\rangle
 = \frac{\theta_{23}\det(v_{1},v_{2},v_{3})}{2\sin\theta_{23}}
\end{equation}
\end{prop}
\begin{proof}
 Fix $y\in T_{4}\setminus\partial T_{4}$.  By definition of $z_{4}$,
\begin{equation}\label{one0.4}
\langle z_{4},y\rangle=\int_{T_{4}}\langle x,y\rangle d\sigma(x).
\end{equation}
Let $H\subset\R^{3}$ be the convex hull of $T_{4}\cup\{0\}$ and let $\{U_{i}\}_{i=1}^{3}\subset\R^{3}$ be half spaces through the origin such that $\{x\in\R^{3}\colon\vnorm{x}_{2}\leq 1\}\bigcap\left(\bigcap_{i=1}^{3}U_{i}\right)=H$; see Figure ~\ref{fig2}.
 \begin{figure}[htbp!]
   \centering
   \def\svgwidth{\figtwowidth} 
   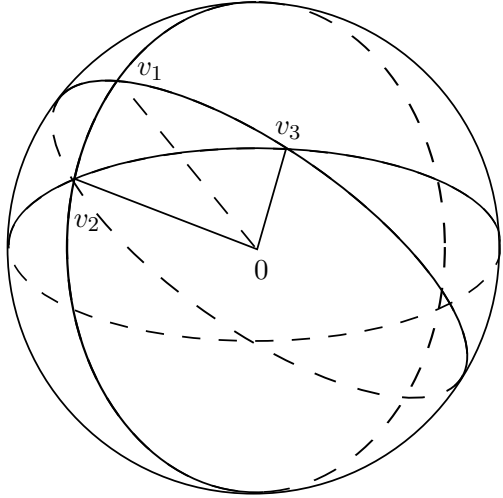
\caption{Spherical triangle together with three disc sectors, each containing $\{v_{i},v_{j},0\}$.}
\label{fig2}
\end{figure}

 Define $\Pi\colon S^{2}\to\R^{2}$ so that $\Pi(x)$ is the orthogonal projection of $x\in S^{2}\subset\R^{3}$ onto the unique plane that is intersecting the origin and perpendicular to $y$.  Applying the Cauchy projection formula (see for example \cite[p.~25]{Kol05}) to $T_{4}$, we see that
\begin{equation}\label{one0.5}
\int_{T_{4}}\langle x,y\rangle d\sigma(x)
=\mbox{Area}_{\R^{2}}\Pi(T_{4}\cap\{x\in S^{2}\colon \langle y,x\rangle>0\})
-\mbox{Area}_{\R^{2}}\Pi(T_{4}\cap\{x\in S^{2}\colon \langle y,x\rangle\leq0\}).
\end{equation}

Define $A\colonequals (T_{4}\setminus\partial T_{4})\cap\{x\in S^{2}\colon \langle y,x\rangle>0\}$ and $B\colonequals(\partial H)\setminus A$.  Note that $\partial H$ is the union of $T_{4}$ and three sectors of discs $\{D_{i}\}_{i=1}^{3}$ such that $D_{i}\subset\partial U_{i}$.  Since $y\in\bigcap_{i=1}^{3}U_{i}$, the exterior normal $n(x)$ of $x\in D_{i}\subset\partial H$ has nonpositive inner product with $y$ for $i=1,2,3$.  So, $\langle n(x),y\rangle>0$ for $x\in\partial H$ if and only if $x\in A$.  Since $H$ is convex, $\Pi\colon A\to\Pi(\partial H)$ and $\Pi\colon B\to\Pi(\partial H)$ are bijections almost everywhere.  Since $\partial H=T_{4}\bigcup\left(\bigcup_{i=1}^{3}D_{i}\right)$, the Cauchy projection formula applied to right side of the identity $\mbox{Area}_{\R^{2}}(\Pi(A))=\mbox{Area}_{\R^{2}}(\Pi(B))$ gives
\begin{equation}\label{one0.6}
\begin{aligned}
&\mbox{Area}_{\R^{2}}\Pi(T_{4}\cap\{x\in S^{2}\colon \langle y,x\rangle>0\})
=\mbox{Area}_{\R^{2}}\Pi(T_{4}\cap\{x\in S^{2}\colon \langle y,x\rangle\leq0\})\\
&\qquad+\frac{\theta_{12}}{2}\left\langle y,\frac{v_{1}\times v_{2}}{\vnorm{v_{1}\times v_{2}}_{2}}\right\rangle
+\frac{\theta_{23}}{2}\left\langle y,\frac{v_{2}\times v_{3}}{\vnorm{v_{2}\times v_{3}}_{2}}\right\rangle
+\frac{\theta_{13}}{2}\left\langle y,\frac{v_{3}\times v_{1}}{\vnorm{v_{3}\times v_{1}}_{2}}\right\rangle.
\end{aligned}
\end{equation}
Since \eqref{one0.6} holds for any $y\in S^{2}\cap T_{4}$, \eqref{one0.4}, \eqref{one0.5} and \eqref{one0.6} give \eqref{onec}.
\end{proof}

Note that the vectors in the determinant of \eqref{oneb} are assumed to have positive orientation.  So, when we apply this equation below we will need to respect orientations.

\begin{cor}\label{cor0.5}
Under the assumptions of Proposition \ref{prop1.5}, we have
\begin{eqnarray}
\label{onef}\|z_4\|_2^2&=&\frac{1}{4}\left(\theta_{12}^2 +\theta_{23}^2 +\theta_{13}^2
-2\cos(\Theta_{213})\theta_{12}\theta_{13}-2\cos(\Theta_{123})\theta_{12}\theta_{23}-2\cos(\Theta_{231})\theta_{23}\theta_{13}\right)\\
&=&
\frac{1}{4}\left(\theta_{12}^2 +\theta_{23}^2 +\theta_{31}^2 +2\frac{\cos\theta_{12}\cos\theta_{13}-\cos\theta_{23}}{\sin\theta_{12}\sin\theta_{13}}\theta_{12}\theta_{13}\right.\nonumber\\
&&\left.\quad+2\frac{\cos\theta_{12}\cos\theta_{23}-\cos\theta_{13}}{\sin\theta_{12}\sin\theta_{23}}\theta_{12}\theta_{23}
+2\frac{\cos\theta_{23}\cos\theta_{13}-\cos\theta_{12}}{\sin\theta_{23}\sin\theta_{13}}\theta_{23}\theta_{13}
\right).\label{onee}
\end{eqnarray}

\end{cor}
\begin{proof}
Apply \eqref{onec}, and then use either the vector identity $\langle a\times b,c\times d\rangle=\langle a,c\rangle\langle b,d\rangle -\langle a,d\rangle\langle b,d\rangle$, or the spherical law of cosines.
\end{proof}

\begin{cor}\label{cor0.6}
For distinct $i,j\in \{1,2,3,4\}$ we have $\langle z_{i},v_{j}\rangle>0$.
\end{cor}
\begin{proof}
Let $k,\ell$ be such that $\det(v_{k},v_{j},v_{\ell})>0$, and apply~\eqref{oneb} and Proposition \ref{prop1}(e).
\end{proof}



\begin{lemma}\label{prop3}
 Write $\langle z_{4}, v_{1}\rangle\equalscolon a$, $\langle z_{4},v_{2}\rangle\equalscolon b$, $\langle z_{4},v_{3}\rangle\equalscolon c$ and $\langle z_{4},v_{4}\rangle\equalscolon-3d$.  Then the following relation holds
\begin{equation}\label{one1}
\frac{1}{a}v_{1}+\frac{1}{b}v_{2}+\frac{1}{c}v_{3}+\frac{1}{d}v_{4}=0
\end{equation}
\end{lemma}
\begin{proof}
By Proposition \ref{prop1}(e) we have $\det(v_{1},v_{2},v_{3})\neq0$, so we may write $v_{4}=\sum_{i=1}^{3}\alpha_{i}v_{i}$ for some $\alpha_1,\alpha_2,\alpha_3\in \R$.  By Proposition \ref{prop1}(c),
$$d=\langle z_{1},v_{4}\rangle=\langle z_{2},v_{4}\rangle=\langle z_{3},v_{4}\rangle.$$
Substituting in the expression $v_{4}=\sum_{i=1}^{3}\alpha_{i}v_{i}$ and then applying Proposition \ref{prop1}(b) and (c) several times, we get
$$
d = -3a\alpha_1 +b\alpha_{2} +c\alpha_{3}
= a\alpha_1 -3b\alpha_{2} +c\alpha_{3}
= a\alpha_1 +b\alpha_{2} -3c\alpha_{3}.
$$
Solving this system of equations gives
$$d=-a\alpha_{1}=-b\alpha_{2}=-c\alpha_{3}.$$
By Corollary \ref{cor0.6}, we may write $\alpha_{1}=-d/a,\alpha_{2}=-d/b,\alpha_{3}=-d/c$.  Substituting these equalities into the expression $v_{4}=\sum_{i=1}^{3}\alpha_{i}v_{i}$ gives \eqref{one1}.  Note that at least one $\alpha_{i}$ is nonzero, so $d$ is nonzero as well, that is, we have not made a division by zero.
\end{proof}

Since $\det(v_{1},v_{2},v_{3})\neq0$, we have $\vnormf{\frac{1}{a}v_{1}+\frac{1}{b}v_{2}+\frac{1}{c}v_{3}}_{2}\neq0$.  Since $\vnorm{v_{4}}_{2}=1$, Lemma \ref{prop3} implies
\begin{equation}\label{one1.4}
v_{4}=-\frac{\frac{1}{a}v_{1}+\frac{1}{b}v_{2}+\frac{1}{c}v_{3}}{\vnormf{\frac{1}{a}v_{1}+\frac{1}{b}v_{2}+\frac{1}{c}v_{3}}_{2}}.
\end{equation}
A priori, we have only determined $v_{4}$ up to multiplication by $\pm1$.  However, Proposition \ref{prop1}(c) and Corollary \ref{cor0.6} imply that $d=(-1/3)\langle z_{4},v_{4}\rangle>0$.  So, taking the inner product of $z_{4}$ with both sides of \eqref{one1.4}  shows that \eqref{one1.4} has the correct sign.

Applying \eqref{oneb} to \eqref{one1.4} gives
\begin{equation}\label{one1.5}
v_{4}=-\frac{\frac{\sin\theta_{23}}{\theta_{23}}v_{1}+\frac{\sin\theta_{13}}{\theta_{13}}v_{2}+\frac{\sin\theta_{12}}{\theta_{12}}v_{3}}
{\vnorm{\frac{\sin\theta_{23}}{\theta_{23}}v_{1}+\frac{\sin\theta_{13}}{\theta_{13}}v_{2}+\frac{\sin\theta_{12}}{\theta_{12}}v_{3}}_{2}}.
\end{equation}

Define, as we did in Section~\ref{sec:overview},
\begin{equation}\label{one1.6}
\lambda
\colonequals\vnorm{\frac{\sin\theta_{23}}{\theta_{23}}v_{1}+\frac{\sin\theta_{13}}{\theta_{13}}v_{2}+\frac{\sin\theta_{12}}{\theta_{12}}v_{3}}_{2}^{2}
\qquad\qquad\qquad\qquad
\end{equation}
\begin{equation}\label{one1.7}
\begin{split}
& = \frac{\sin^{2}\theta_{23}}{\theta_{23}^{2}}+\frac{\sin^{2}\theta_{13}}{\theta_{13}^{2}}+\frac{\sin^{2}\theta_{12}}{\theta_{12}^{2}}
+2\frac{\sin\theta_{23}\sin\theta_{13}}{\theta_{23}\theta_{13}}\cos\theta_{12}\\
&\qquad+2\frac{\sin\theta_{23}\sin\theta_{12}}{\theta_{23}\theta_{12}}\cos\theta_{13}
+2\frac{\sin\theta_{12}\sin\theta_{13}}{\theta_{12}\theta_{13}}\cos\theta_{23}.
\end{split}
\end{equation}
%
%

\begin{prop}\label{prop5}
\textup{(a)} All four vertices of $P$ are not contained in any closed half sphere. 
\begin{itemize}
\item[(b)]  The function $F(P)$ is well-defined as a function of the vertices $\{v_{i}\}_{i=1}^{4}$.
\end{itemize}
\end{prop}
\begin{proof}
(a) By Proposition \ref{prop1}(e) we have $0<\theta_{ij}<\pi$.  \eqref{one1.5} says that $v_{4}$ is in the complement of all closed half spheres containing $\{v_{1},v_{2},v_{3}\}$.

(b) From Proposition \ref{prop1}(e), the edges $T_{i}\cap T_{j}$ of $P$ are geodesics of length less than $\pi$, so the vertices uniquely determine the edges.  Finally, the edges determine $P$ itself.
\end{proof}
%
%

\begin{prop}[Restrictions on products of opposite edges]\label{prop6}
\begin{equation}\label{one2}
\frac{\sin\theta_{12}\sin\theta_{34}}{\theta_{12}\theta_{34}}
=\frac{\sin\theta_{13}\sin\theta_{24}}{\theta_{13}\theta_{24}}
=\frac{\sin\theta_{14}\sin\theta_{23}}{\theta_{14}\theta_{23}}.
\end{equation}
\end{prop}
\begin{proof}
From Proposition \ref{prop1}(c) we have $\langle z_{3},v_{1}\rangle=\langle z_{2},v_{1}\rangle$.  For each side of the equality $\langle z_{3},v_{1}\rangle=\langle z_{2},v_{1}\rangle$, apply \eqref{oneb} for $\langle z_{i},v_{1}\rangle$ and then substitute in \eqref{one1.5} to get
$$-\frac{\det(v_{1},v_{3},v_{2})}{2\sqrt{\lambda}}\frac{\sin\theta_{12}}{\theta_{12}}\frac{\theta_{24}}{\sin\theta_{24}}
=-\frac{\det(v_{1},v_{3},v_{2})}{2\sqrt{\lambda}}\frac{\sin\theta_{13}}{\theta_{13}}\frac{\theta_{34}}{\sin\theta_{34}}.$$
Canceling terms gives the first equality of \eqref{one2}.  The second follows similarly: Proposition \ref{prop1}(c) says $\langle z_{3},v_{2}\rangle=\langle z_{1},v_{2}\rangle$, and so on.
\end{proof}


\begin{lemma}[Restrictions on a single triangle]\label{prop7}
  Let $\lambda$ be defined as in \eqref{one1.7} and let $\gamma$ be defined as in~\eqref{one3.00}, i.e.,
\begin{equation}\label{one3.0}
\gamma\colonequals\frac{\sin\theta_{23}}
{\theta_{23}}\cos\theta_{12}+\frac{\sin\theta_{13}}{\theta_{13}}+\frac{\sin\theta_{12}}{\theta_{12}}\cos\theta_{23}.
\end{equation}
Then the following holds
\begin{equation}\label{one3}
\cos\left(\sqrt{\lambda-\gamma^{2}}\frac{\theta_{23}\theta_{12}}{\sin\theta_{23}\sin\theta_{12}}\right)=-\frac{\gamma}{\sqrt{\lambda}}.
\end{equation}
\begin{equation}\label{one5}
\cos\left(\sqrt{\theta_{12}^{2}+\theta_{23}^{2}+2\theta_{12}\theta_{23}\cos\Theta_{123}}\right)=-\frac{\gamma}{\sqrt{\lambda}}.
\end{equation}
\end{lemma}

\begin{proof}
Arguing as in Proposition \ref{prop6}, we start with Proposition \ref{prop1}(c) and the equality
$\langle z_{4},v_{1}\rangle=\langle z_{3},v_{1}\rangle$.
For each side of this equality, apply \eqref{oneb} for $\langle z_{i},v_{1}\rangle$ and substitute in \eqref{one1.5} to get
\begin{equation}\label{one4}
\frac{\theta_{23}\theta_{12}}{\sin\theta_{23}\sin\theta_{12}}=\frac{\theta_{24}}{\sin\theta_{24}}\frac{1}{\sqrt{\lambda}}.
\end{equation}
  Substituting  $\theta_{24}=\arccos(\langle v_{2},v_{4}\rangle)$ into \eqref{one4} and rearranging gives
$$
\cos\left(\sqrt{\lambda}\sqrt{1-\langle v_{2},v_{4}\rangle^{2}}\frac{\theta_{23}\theta_{12}}{\sin\theta_{23}\sin\theta_{12}}\right)=\langle v_{2},v_{4}\rangle.
$$
From \eqref{one1.5} and the definitions of $\lambda$ and $\gamma$ (\eqref{one1.7} and \eqref{one3.0}), we have $\langle v_{2}, v_{4}\rangle=-\gamma/\sqrt{\lambda}$, so
$$
\cos\left(\sqrt{\lambda}\sqrt{1-\frac{\gamma^{2}}
{\lambda}}\frac{\theta_{23}\theta_{12}}{\sin\theta_{23}\sin\theta_{12}}\right)=-\frac{\gamma}{\sqrt{\lambda}}.
$$
yielding \eqref{one3}.  To derive \eqref{one5}, observe that $\lambda-\gamma^{2}$ allows some cancelation as follows.
\begin{flalign*}
\lambda-\gamma^{2}
& = \frac{\sin^{2}\theta_{23}}{\theta_{23}^{2}}(1-\cos^{2}\theta_{12})
+\frac{\sin^{2}\theta_{12}}{\theta_{12}^{2}}(1-\cos^{2}\theta_{23})
+2\frac{\sin\theta_{23}\sin\theta_{12}}{\theta_{23}\theta_{12}}(\cos\theta_{31}-\cos\theta_{23}\cos\theta_{12})\\
& = \frac{\sin^{2}\theta_{23}}{\theta_{23}^{2}}(\sin^{2}\theta_{12})
+\frac{\sin^{2}\theta_{12}}{\theta_{12}^{2}}(\sin^{2}\theta_{23})
+2\frac{\sin\theta_{23}\sin\theta_{12}}{\theta_{23}\theta_{12}}(\cos\theta_{31}-\cos\theta_{23}\cos\theta_{12}).
\end{flalign*}
Thus,
\begin{flalign*}
(\lambda-\gamma^{2})\frac{\theta_{23}^{2}\theta_{12}^{2}}{\sin^{2}\theta_{23}\sin^{2}\theta_{12}}
& = \theta_{12}^{2}+\theta_{23}^{2}+2\theta_{12}\theta_{23}\frac{\cos\theta_{31}-\cos\theta_{23}\cos\theta_{12}}{\sin\theta_{12}\sin\theta_{23}}\\
& = \theta_{12}^{2}+\theta_{23}^{2}+2\theta_{12}\theta_{23}\cos\Theta_{123}\quad\mbox{, by the spherical law of cosines.}
\end{flalign*}
Substituting this equation into \eqref{one3} gives \eqref{one5}.
\end{proof}

Since $-\frac{\gamma}{\sqrt{\lambda}}=\langle v_{4},v_{2}\rangle=\cos\theta_{24}$, \eqref{one5} says that
\begin{equation}\label{one6}
\theta_{12}^{2}+\theta_{23}^{2}+2\theta_{12}\theta_{23}\cos\Theta_{123}=\theta_{24}^{2}.
\end{equation}
Substituting \eqref{one6} back into \eqref{onef} and \eqref{eq:the F} gives a simplified form of $F(P)$.
\begin{lemma}\label{lem:squares}
\begin{equation}\label{one7}
F(P)=\sum_{i<j}\theta_{ij}^{2}=\sum_{i<j}(\arccos(\langle v_{i},v_{j}\rangle))^{2}
\end{equation}
\end{lemma}
\begin{proof}
Applying \eqref{onee} to the definition $F(P)=\sum_{i=1}^{4}\vnorm{z_{i}}_{2}^{2}$ gives
\begin{equation}\label{one7.1}
F(P)=\frac{1}{2}\sum_{i<j}\theta_{ij}^{2}-\frac{1}{2}\sum_{i,j,\ell}\cos(\Theta_{ij\ell})\theta_{ij}\theta_{j\ell}.
\end{equation}
Here the first sum in~\eqref{one7.1} runs over all $(i,j)$ satisfying $1\leq i<j\leq 4$, and the second sum runs over all equivalence classes of $3$-tuples $(i,j,\ell)$ under the equivalence $(i,j,\ell)\sim(\ell,j,i)$, where $i,j,\ell$ are distinct elements of $\{1,2,3,4\}$.  From \eqref{one6} we have, for $r\notin\{i,j,\ell\}$,
\begin{equation}\label{one7.2}
\theta_{ij}^{2}+\theta_{j\ell}^{2}+2\theta_{ij}\theta_{j\ell}\cos\Theta_{ij\ell}=\theta_{jr}^{2}.
\end{equation}
Substituting~\eqref{one7.2} into \eqref{one7.1} gives
$$
F(P)=\frac{1}{2}\sum_{i<j}\theta_{ij}^{2}-\frac{1}{4}\sum_{i,j,\ell}\left(\theta_{jr}^{2}-\theta_{ij}^{2}-\theta_{j\ell}^{2}\right).
$$
In the right summation, for a fixed $(i',j')$ with $i'\neq j'$, the term $\theta_{i'j'}^{2}$ appears four times with a minus sign and two times with a plus sign.  Therefore,
$$
F(P)
=\frac{1}{2}\sum_{i<j}\theta_{ij}^{2}-\frac{1}{4}\sum_{i,j,\ell}\left(\theta_{jr}^{2}-\theta_{ij}^{2}-\theta_{j\ell}^{2}\right)
=\frac{1}{2}\sum_{i<j}\theta_{ij}^{2}-\frac{1}{4}(\sum_{i<j}-2\theta_{ij}^{2})
=\sum_{i<j}\theta_{ij}^{2},
$$
proving \eqref{one7}.
\end{proof}


In \eqref{one7} replace $\theta_{jr}^{2}$ with the left side of \eqref{one7.2} for $(j,r)\in\{(1,4),(2,4),(3,4)\}$ to get
\begin{equation}\label{one8}
F(P)=3\left(\theta_{12}^{2}+\theta_{23}^{2}+\theta_{13}^{2}\right)
+2\cos(\Theta_{213})\theta_{12}\theta_{13}+2\cos(\Theta_{123})\theta_{12}\theta_{23}+2\cos(\Theta_{231})\theta_{23}\theta_{13}.
\end{equation}
\begin{definition}\label{F0def}
By the spherical law of cosines, $\Theta_{ij\ell}$ is a function of $\{\theta_{ij},\theta_{j\ell},\theta_{i\ell}\}$.  Let
$$
F_{0}(\theta_{12},\theta_{23},\theta_{13})\colonequals
3\left(\theta_{12}^{2}+\theta_{23}^{2}+\theta_{13}^{2}\right)
+2\cos(\Theta_{213})\theta_{12}\theta_{13}+2\cos(\Theta_{123})\theta_{12}\theta_{23}+2\cos(\Theta_{231})\theta_{23}\theta_{13}.
$$
\end{definition}
By \eqref{one8}, $F(P)=F_{0}$ can be computed by the data from only one triangle of $P$.  That is,
$$
F_{0}(\theta_{12},\theta_{23},\theta_{13})
=F_{0}(\theta_{12},\theta_{24},\theta_{14})
=F_{0}(\theta_{34},\theta_{24},\theta_{23})
=F_{0}(\theta_{34},\theta_{14},\theta_{13}).
$$

Since the spherical propeller partition, i.e., the partition of $S^2$ corresponding to the intersection of the propeller partition of $\R^3$ with $S^2$, satisfies $F=(9/4)\pi^{2}$, we may bound \eqref{one8} in terms of $M\colonequals\max\{\theta_{12},\theta_{13},\theta_{23}\}$.  In particular, we get $F(P)\leq 14 M^{2}$, which is less than $(9/4)\pi^{2}$ if $M<\pi\frac{3}{2\sqrt{14}}$.  (In our bound on $F(P)$, we use that one $\Theta_{ijk}$ must be larger than $\pi/3$ by the pigeonhole principle applied to $\Theta_{123}+\Theta_{231}+\Theta_{312}-\pi=\mbox{Area}_{S^{2}}(T_{4})>0$, so one cosine term satisfies $\cos\Theta_{ijk}\leq1/2$.)  So, this bound on $F$  gives
\begin{cor}\label{cor2.5}
Any triangle $T_{i}$ of $P$ must have an edge of length greater than $\pi\frac{3}{2\sqrt{14}}>\frac{5}{4}$.
\end{cor}

Iteratively applying \eqref{one2} gives the following improvement to Proposition \ref{prop1}(e).  This improvement will be important in our numerical calculations in Section \ref{secnet} and in the proof of Theorem~\ref{thm:main}, since we eventually require a bound on the derivative of $\theta_{ij}/\sin\theta_{ij}$.
\begin{lemma}[Restrictions on edge length]\label{lemma4}  For every distinct $i,j\in \{1,2,3,4\}$,
\begin{equation}\label{three2}
\theta_{ij}<\pi-\frac12.
\end{equation}
\end{lemma}
\begin{proof}
Suppose $\theta_{12}\geq\pi-1/2$.  We are eventually going to derive a contradiction, which will allow us to conclude that $\theta_{12}<\pi-1/2$.  We will essentially only need \eqref{one2}, the monotonicity of $x\mapsto(\sin x)/x$ for $x\in(0,\pi)$, the fact that the perimeter $\theta_{ij}+\theta_{j\ell}+\theta_{\ell i}$ of a spherical triangle is bounded by $2\pi$, and that $(\sin x)/x\leq1$ on $[0,\pi]$.  Relabeling the edges if necessary, we may assume that $\theta_{13}=\max\{\theta_{13},\theta_{24}\}$.  Using $\theta_{12}\geq\pi-1/2$, \eqref{one2} gives
\begin{equation}\label{three2.1}
\frac{\sin(\pi-1/2)}{\pi-1/2}
\geq\frac{\sin\theta_{12}}{\theta_{12}}
\geq\frac{\sin\theta_{12}\sin\theta_{34}}{\theta_{12}\theta_{34}}
=\frac{\sin\theta_{13}\sin\theta_{24}}{\theta_{13}\theta_{24}}
\geq\left(\frac{\sin\theta_{13}}{\theta_{13}}\right)^{2}.
\end{equation}
Therefore, $\theta_{13}\geq X$ where $X\in(0,\pi)$ satisfies
\begin{equation}\label{three2.2}
\left(\frac{\sin X}{X}\right)^{2}=\frac{\sin(\pi-1/2)}{\pi-1/2}.
\end{equation}
So, $\theta_{13}\geq X>2.065$.  Now, since $\theta_{12}+\theta_{13}+\theta_{23}\leq2\pi$ the bounds on $\theta_{12}$ and $\theta_{13}$ give
$$\theta_{23}\leq 2\pi-(\pi-1/2+X)\equalscolon Y<1.577.$$
But then \eqref{one2} and the bounds on $\theta_{12}$ and $\theta_{23}$ yield
$$
\frac{\sin\theta_{14}}{\theta_{14}}
=\frac{\sin\theta_{12}\sin\theta_{34}}{\theta_{12}\theta_{34}}\frac{\theta_{23}}{\sin\theta_{23}}
\leq\frac{\sin(\pi-1/2)}{\pi-1/2}\frac{Y}{\sin Y}.
$$
So, $\theta_{14}\geq Z$ where $Z\in(0,\pi)$ is defined by
$$\frac{\sin Z}{Z}=\frac{\sin(\pi-1/2)}{\pi-1/2}\frac{Y}{\sin Y}.$$
Then $\theta_{14}\geq Z>2.388$.  But since $\theta_{14}+\theta_{24}+\theta_{12}\leq2\pi$ the bounds on $\theta_{14}$ and $\theta_{12}$ give
$$\theta_{24}\leq2\pi-(\pi-1/2+Z)\equalscolon W<1.254$$

We now iterate the above argument.  Using \eqref{three2.1},
\begin{equation}\label{three2.3}
\frac{\sin\theta_{13}}{\theta_{13}}
=\frac{\sin\theta_{12}\sin\theta_{34}}{\theta_{12}\theta_{34}}\frac{\theta_{24}}{\sin\theta_{24}}
\leq\frac{\sin(\pi-1/2)}{\pi-1/2}\frac{W}{\sin W}.
\end{equation}
So, $\theta_{13}\geq\tilde{X}$ with $(\sin\tilde{X})/\tilde{X}=(\sin(\pi-1/2)/(\pi-1/2))(W/\sin W)$.  We now repeat the argument above, starting at \eqref{three2.2} and using tildes to designate updated variables.  Thus, $\theta_{13}\geq\tilde{X}>2.499$, $\theta_{23}\leq\tilde{Y}< 1.14, \theta_{14}\geq\tilde{Z}> 2.52,\theta_{24}\leq\tilde{W}< 1.122$.  Using $\tilde{W}$ in place of $W$ in \eqref{three2.3}, we find that $\theta_{13}>2.52$.

In summary, iterating the above procedure improves our bounds on the $\theta_{ij}$.  We find that $\theta_{13},\theta_{14}>2.52$.  We also find that $\theta_{23},\theta_{24}<1.14$.  Since $\theta_{13}+\theta_{14}+\theta_{34}\leq2\pi$, $\theta_{34}\leq2\pi-(2(2.52))$.  But then $\theta_{23},\theta_{24},\theta_{34}<5/4<\pi\frac{3}{2\sqrt{14}}$, violating Corollary \ref{cor2.5} and giving our desired contradiction.  We therefore conclude that $\theta_{ij}<\pi-1/2$.
\end{proof}

%

For a partition $P=\{T_{i}\}_{i=1}^{4}$ as above, the following map $G\colon(S^{2})^{4}\to\R$ is also a function of $P$ by Proposition \ref{prop5}(b)
\begin{equation}\label{one9}
G(v_{1},v_2,v_3,v_{4})\colonequals\sum_{1\leq i<j\leq4}\theta_{ij}^{2}.
\end{equation}
\begin{prop}\label{prop8}
Suppose that the conclusion of Proposition \ref{prop1}(e) holds.  Then Proposition \ref{prop1}(c) holds if and only if $P$ is a critical point of $G$.
\end{prop}
\begin{proof}
For a manifold $M$ with $v\in M$, $T_{v}M$ denotes the tangent space of $M$ at $v$.  Since $G=\sum_{i<j}\theta_{ij}^{2}=\sum_{i<j}(\arccos\langle v_{i},v_{j}\rangle)^{2}$, we take derivatives of $G$ with respect to $v_{i}$.  For $v_{2},v_{3},v_{4}\in S^{2}$ fixed and $x\in \R^{3}$ variable,
$$\nabla_{S^{2}}G(\cdot,v_{2},v_{3},v_{4})=\left.\left(\nabla_{\R^{3}}G\left(\frac{x}{\vnorm{x}_{2}},v_{2},v_{3},v_{4}\right)\right)\right|_{S^{2}}$$
Recall that $T_{v_{1}}\R^{3}$ is isomorphic to $\R^{3}$.  We denote this isomorphism by $(d/dy)\leftrightarrow y$.

Let $\frac{d}{dy}$ be a vector in the tangent space $T_{v_{1}}\R^{3}$.  Identifying $\frac{d}{dy}$ with $y\in\R^{3}$ we compute
\begin{flalign*}
&\frac{d}{dy}G\left(\frac{x}{\vnorm{x}_{2}},v_{2},v_{3},v_{4}\right)
 = \sum_{j\neq1}\frac{-2\arccos(\langle x,v_{j}\rangle/\vnorm{x}_{2})}{\sqrt{1-(\langle x,v_{j}\rangle/\vnorm{x}_{2})^{2}}}\left\langle\frac{d}{dy}(x/\vnorm{x}_{2}),v_{j}\right\rangle\\
&\qquad\qquad\qquad = \sum_{j\neq1}\frac{-2\arccos(\langle x,v_{j}\rangle/\vnorm{x}_{2})}{\sqrt{1-(\langle x,v_{j}\rangle/\vnorm{x}_{2})^{2}}}
\left(\frac{\vnorm{x}_{2}\frac{d}{dy}\langle x,v_{j}\rangle-\frac{1}{2}\langle x,v_{j}\rangle\vnorm{x}_{2}^{-1}\frac{d}{dy}\langle x,x\rangle}{\vnorm{x}_{2}^{2}}\right)\\
&\qquad\qquad\qquad = \sum_{j\neq1}\frac{-2\arccos(\langle x,v_{j}\rangle/\vnorm{x}_{2})}{\sqrt{1-(\langle x,v_{j}\rangle/\vnorm{x}_{2})^{2}}}
\left\langle\left(\frac{\vnorm{x}_{2}y-x\vnorm{x}_{2}^{-1}\langle y,x\rangle}{\vnorm{x}_{2}^{2}}\right),v_{j}\right\rangle.
\end{flalign*}
Letting $\frac{d}{dy}\in T_{v_{1}}S^{2}$ and $x=v_{1}\in S^{2}$ gives $\langle y,x\rangle=0$, so that
$$
\frac{d}{dy}G\left(\frac{x}{\vnorm{x}_{2}},v_{2},v_{3},v_{4}\right)
=\sum_{j\neq1}\frac{-2\arccos(\langle v_{1},v_{j}\rangle)}{\sqrt{1-(\langle v_{1},v_{j}\rangle)^{2}}}
\left\langle y, v_{j}\right\rangle
=\sum_{j\neq1}\frac{-2\theta_{1j}}{\sin\theta_{1j}}\left\langle y,v_{j}\right\rangle.
$$

So, if we take $\frac{d}{dy}\colonequals v_{1}\times v_{2}\in T_{v_{1}}S^{2}$ and then use \eqref{oneb}, we get
\begin{flalign*}
\frac{d}{dy}G
 = \frac{-2\theta_{13}}{\sin\theta_{13}}\det(v_{1},v_{2},v_{3})-\frac{2\theta_{14}}{\sin\theta_{14}}\det(v_{1},v_{2},v_{4})
 = 2(-\langle z_{4},v_{2}\rangle+\langle z_{3},v_{2}\rangle).
\end{flalign*}


In general, if we take the derivative at $v_{i}$ in the direction $v_{i}\times v_{j}$ and set this derivative to zero, we get $\langle z_{r},v_{j}\rangle=\langle z_{\ell},v_{j}\rangle$ where the set $\{i,j,\ell,r\}$ is equal to the set $\{1,2,3,4\}$.  Since the above argument can be reversed, the proposition is proven.
\end{proof}

Setting $\frac{d}{dy}=v_{2}-\langle v_{2},v_{1}\rangle v_{1}$ and $\frac{d}{dy}=v_{3}-\langle v_{3},v_{1}\rangle v_{3}$ in the above argument gives
\begin{equation}\label{one10}
-\theta_{12}=\theta_{13}\cos\Theta_{213}+\theta_{14}\cos\Theta_{214}.
\end{equation}

Observe that
\begin{flalign*}
&\frac{d}{dy}G\left(\frac{x}{\vnorm{x}_{2}},v_{2},v_{3},v_{4}\right)
 = \sum_{j\neq1}\frac{-2\theta_{1j}}{\sin\theta_{1j}}\left\langle y,v_{j}\right\rangle
= \sum_{j\neq1}\frac{-2\theta_{1j}}{\sin\theta_{1j}}\left\langle(v_{2}-\langle v_{1},v_{2}\rangle v_{1}),v_{j}\right\rangle\\
&\quad = -\frac{2\theta_{12}}{\sin\theta_{12}}(1-(\langle v_{2},v_{1}\rangle)^{2})
-\frac{2\theta_{13}}{\sin\theta_{13}}(\langle v_{2},v_{3}\rangle-(\langle v_{2},v_{1}\rangle)(\langle v_{1},v_{3}\rangle))\\
&\qquad\qquad-\frac{2\theta_{14}}{\sin\theta_{14}}(\langle v_{2},v_{4}\rangle-(\langle v_{2},v_{1}\rangle)(\langle v_{1},v_{4}\rangle))\\
&\quad = -\frac{2\theta_{12}}{\sin\theta_{12}}(1-(\langle v_{2},v_{1}\rangle)^{2})
-\frac{2\theta_{13}}{\sin\theta_{13}}\left\langle(v_{1}\times v_{2}),(v_{1}\times v_{3})\right\rangle
-\frac{2\theta_{14}}{\sin\theta_{14}}\left\langle(v_{1}\times v_{2}),(v_{1}\times v_{4})\right\rangle.
\end{flalign*}
Thus, dividing by $\sin\theta_{12}$, setting $\frac{d}{dy}G=0$, and simplifying as in Corollary \ref{cor0.5} gives \eqref{one10}.

We now give a geometric interpretation of \eqref{one6} and \eqref{one10} that should clarify their meaning.  Define $\exp_{1}^{-1}\colon\{2,3,4\}\to \R^{2}$ by
\begin{equation}\label{one9.9}
\exp_{1}^{-1}(j)\colonequals\theta_{1j}\frac{(v_{j}-\langle v_{j},v_{1}\rangle v_{1})}{\vnorm{v_{j}-\langle v_{j},v_{1}\rangle v_{1}}_{2}}.
\end{equation}
Then $\exp_{1}^{-1}(j)$ is perpendicular to $v_{1}$ and parallel to $v_{j}-\mbox{Proj}_{v_{1}}v_{j}$, with length $\theta_{1j}$.
\begin{lemma}\label{lemma5} Suppose that the conclusion of Proposition \ref{prop1}(e) holds.  If $\{v_{i}\}_{i=1}^{4}$ is a critical point of $G$, then
\begin{equation}\label{one10prime}
\exp_{1}^{-1}(2)+\exp_{1}^{-1}(3)+\exp_{1}^{-1}(4)=0.
\end{equation}
\end{lemma}
\begin{proof}
Let $\Pi_{ij}\subset\R^{3}$ be the plane containing $\{0,v_{i},v_{j}\}$.  For $j\in\{2,3,4\}$, $\exp_{1}^{-1}(j)\in\Pi_{1j}$ is perpendicular to $v_{1}$ and in $\Pi_{1j}$, so $\langle\exp_{1}^{-1}(2),\exp_{1}^{-1}(3)\rangle=\theta_{12}\theta_{13}\cos\Theta_{213}$.  We can therefore rewrite \eqref{one6} as
\begin{equation}\label{one6prime}
\vnorm{\exp_{1}^{-1}(2)+\exp_{1}^{-1}(3)}_{2}^{2}=\vnorm{\exp_{1}^{-1}(4)}_{2}^{2}.
\end{equation}
Using the definition of $\exp_{1}^{-1}(j)$, we similarly rewrite \eqref{one10} (with its indices permuted) as
\begin{equation}\label{one9prime}
-\exp_{1}^{-1}(4)=\mbox{Proj}_{\exp_{1}^{-1}(4)}(\exp_{1}^{-1}(2))+\mbox{Proj}_{\exp_{1}^{-1}(4)}(\exp_{1}^{-1}(3)).
\end{equation}
Since $\exp_{1}^{-1}(j)$ is perpendicular to $v_{1}$ for $j\in\{2,3,4\}$, $\{\exp_{1}^{-1}(j)\}_{j\in\{2,3,4\}}$ lies in a plane containing the origin.  Combining \eqref{one6prime} and \eqref{one9prime} therefore gives \eqref{one10prime}.
\end{proof}
In particular, orthogonally projecting \eqref{one10prime} orthogonal to $\exp_{1}^{-1}(2)$ gives
\begin{equation}\label{one12}
\theta_{13}\sin\Theta_{213}=\theta_{14}\sin\Theta_{214}.
\end{equation}

The similarity of \eqref{one6},\eqref{one10} and \eqref{one12} is no coincidence.  For example, note that \eqref{one10} and \eqref{one12} express $\sum_{j\in \{2,3,4\}}\exp_{1}^{-1}(j)=0$, tangential and perpendicular to $\exp_{1}^{-1}(2)$, respectively.  So, combining Proposition \ref{prop8}, \eqref{one10}, \eqref{one12} and Lemma \ref{lemma5}, we deduce the equivalence
\begin{cor}\label{cor3}  Suppose that the conclusion of Proposition \ref{prop1}(e) holds.  Then
$$
\left\{\substack{\mbox{Prop. \ref{prop1}(c)}\\\mbox{holds}}\right\}
\Leftrightarrow
\left\{\substack{\mbox{$P$ is a}\\ \mbox{critical point of $G$}}\right\}
\Leftrightarrow
\left\{\substack{\mbox{\eqref{one10prime} holds at each $v_{i}$.  That is,}\\ \mbox{for each $i$, $\sum_{j\colon j\neq i}\exp_{i}^{-1}(j)=0$}}\right\}
$$
\end{cor}

\begin{lemma}\label{thm1}
$\theta_{ij}>1/10$ and $\sin\Theta_{ij\ell}>1/10$.
\end{lemma}
\begin{proof}
After recording a few identities, we consider a few cases in proving $\theta_{ij}>1/10$.  In all cases, we argue by contradiction and assume $0<\theta_{12}\leq1/10$.  Also, by relabeling if necessary, we may assume that $\max\{\theta_{13},\theta_{23},\theta_{14},\theta_{24}\}=\max\{\theta_{13},\theta_{23}\}$.

\begin{equation}\label{four3}
\cos\Theta_{213}=\frac{\cos\theta_{23}-\cos\theta_{12}\cos\theta_{13}}{\sin\theta_{12}\sin\theta_{13}}.
\end{equation}
\begin{equation}\label{four4}
\cos\Theta_{132}=\frac{\cos\theta_{12}-\cos\theta_{13}\cos\theta_{23}}{\sin\theta_{13}\sin\theta_{23}}.
\end{equation}
\begin{equation}\label{four4.5}
\begin{aligned}
\cos\theta_{34}&=\cos\theta_{13}\cos\theta_{14}+\sin\theta_{13}\sin\theta_{14}\cos\Theta_{314}\\
&=\cos\theta_{13}\cos\theta_{14}+\sin\theta_{13}\sin\theta_{14}\cos(2\pi-\Theta_{213}-\Theta_{214}).
\end{aligned}
\end{equation}
\begin{equation}\label{four4.6}
\absf{\arccos(x+y)-\arccos(x)}\leq\arccos(1-y),\quad x,y,x+y\in[-1,1].
\end{equation}

\noindent {\bf (Case 1)} $\pi/2\leq\max\{\theta_{13},\theta_{23}\}\leq\pi/2+1/10$.  Without loss of generality, label the edges so that $\theta_{13}\geq\theta_{23}$, so $\pi/2\leq\theta_{13}\leq\pi/2+1/10$.  Then $\theta_{14}\geq\theta_{24}$, from Proposition \ref{prop1}(b).  The Proposition shows that $\mbox{Proj}_{\mbox{span}\{v_{1},v_{2}\}}(z_{4}-z_{3})=0$.  Suppose $\theta_{14}<\theta_{24}$.  Let $e$ be the midpoint of the edge $T_{4}\cap T_{3}$ and let $n\in\mbox{span}\{v_{1},v_{2}\}$ be perpendicular to $e$ such that $\langle n,v_{2}\rangle>0$.  Since $\theta_{13}\geq\theta_{23}$, $\langle z_{4},n\rangle\geq0$.  To see this, let $\Pi\colon S^{2}\to\R^{2}$ denote projection onto the unique plane intersecting the origin and perpendicular to $n$.  Without loss of generality, the vertices are oriented so that $\det(v_{1},v_{2},v_{3})>0$.  Let $m$ be the point of intersection of this plane with $(\partial T_{4})\setminus(\partial T_{3})$.  The spherical triangle $T_{4}'$ with vertices $\{v_{1},v_{2},m\}$ then has two equal length edges $\{v_{1},m\}$ and $\{v_{2},m\}$, and $T_{4}'\subseteq T_{4}$.  Also, $\theta_{13}\geq\theta_{23}$ implies that $T_{4}\cap\{x\in S^{2}\colon\langle x,n\rangle\leq0\}=T_{4}'\cap\{x\in S^{2}\colon\langle x,n\rangle\leq0\}$.  Since $T_{4}\supseteq T_{4}'$ and $T_{4}'$ is isosceles,
\begin{flalign*}
\Pi\left(T_{4}\cap\{x\in S^{2}\colon\langle x,n\rangle\geq0\}\right)
\supseteq\Pi\left(T_{4}'\cap\{x\in S^{2}\colon\langle x,n\rangle\geq0\}\right)
=\Pi\left(T_{4}\cap\{x\in S^{2}\colon\langle x,n\rangle\leq0\}\right).
\end{flalign*}
So \eqref{one0.5} shows $\langle z_{4},n\rangle\geq0$.  Similarly, $\theta_{14}<\theta_{24}$, implies $\langle z_{3},n\rangle<0$, a contradiction.  Now
$$
\cos\Theta_{213}
=\frac{\cos\theta_{23}-\cos\theta_{13}\cos\theta_{12}}{\sin\theta_{13}\sin\theta_{12}}
\geq\frac{\cos\theta_{13}(1-\cos\theta_{12})}{\sin\theta_{13}\sin\theta_{12}}
=\frac{\tan(\theta_{12}/2)}{\tan\theta_{13}}
\geq\frac{\tan(1/20)}{\tan(\pi/2+1/10)}.
$$
So, $\Theta_{213}\leq\arccos(\frac{\tan(1/20)}{\tan(\pi/2+1/10)})<\pi/2+.01$.  By \eqref{one6}, $(\theta_{13}-\theta_{12})^{2}\leq\theta_{14}^{2}\leq(\theta_{12}+\theta_{13})^{2}$, so $\pi/2-1/10\leq\theta_{14}\leq\theta_{13}+\theta_{12}\leq\pi/2+1/10+1/10$.  So we similarly conclude that $\Theta_{214}\leq\arccos(\frac{\tan(1/20)}{\tan(\pi/2+1/5)})<\pi/2+.02$.
By \eqref{four4.5},
\begin{equation}\label{four4.65}
\cos\theta_{34}=\cos(\theta_{13}+\theta_{14})+\sin\theta_{13}\sin\theta_{14}(1+\cos(2\pi-\Theta_{213}-\Theta_{214})).
\end{equation}
Since $\Theta_{314}<\pi$ by Proposition \ref{prop1}(e), $$\abs{1+\cos\Theta_{314}}=\abs{1+\cos(2\pi-\Theta_{213}-\Theta_{214})}\leq \abs{1+\cos(2\pi-\pi-.03)}<.0005.$$  Then \eqref{four4.6} applied to \eqref{four4.65} gives a dichotomy: $(i)$ if $\theta_{13}+\theta_{14}\leq\pi$, then $\abs{\theta_{34}-(\theta_{13}+\theta_{14})}<.04$.  Or $(ii)$ if $\theta_{13}+\theta_{14}>\pi$, then $\abs{\theta_{34}-(2\pi-\theta_{13}+\theta_{14})}<.04$.  In case $(i)$, $\theta_{34}\geq\pi/2+\pi/2-1/10-.04>\pi-1/2$.  In case $(ii)$, $\theta_{34}\geq2\pi-(\pi/2+1/10+\pi/2+1/5)-.04>\pi-1/2$.  So, in either case, Lemma \ref{lemma4} is violated.  Therefore Case 1 cannot hold.

\noindent{\bf (Case 2)} $\pi/2+1/10\leq\max\{\theta_{13},\theta_{23}\}\leq2.9$.  Without loss of generality, label the edges so that $\theta_{13}\geq\theta_{23}$, so $\pi/2+1/10\leq\theta_{13}\leq2.9$.  By the triangle inequality on the sphere, $\theta_{23}\geq\theta_{13}-\theta_{12}\geq\theta_{13}-1/10\geq\pi/2$, so
$$
\theta_{13}\cos\Theta_{132}
=\theta_{13}\left(\frac{\cos\theta_{12}-\cos\theta_{13}\cos\theta_{23}}{\sin\theta_{13}\sin\theta_{23}}\right)
\geq\theta_{13}\left(\frac{\cos(1/10)-\cos\theta_{13}\cos\theta_{13}}{\sin\theta_{13}\sin(\theta_{13}-1/10)}\right)
>\frac{\pi}{2}.
$$
Since $\theta_{23},\theta_{13}\geq\pi/2$, \eqref{one6} gives
$$
\theta_{34}
=\sqrt{\theta_{13}^{2}+\theta_{23}^{2}+2\theta_{13}\theta_{23}\cos\Theta_{132}}
>\sqrt{2(\pi/2)^{2}+2(\pi/2)(\pi/2)}=\pi,
$$
so \eqref{three2} is violated.  Therefore Case 2 cannot occur.

\noindent{\bf (Case 3)} The case $\max\{\theta_{13},\theta_{23}\}>2.9$ cannot occur, by \eqref{three2}.  Also, $\max\{\theta_{13},\theta_{23},\theta_{14},\theta_{24}\}<\pi/2$ implies that $T_{4}$ and $T_{3}$ are contained in separate but adjacent quarter spheres, so $T_{4}\cup T_{3}$ is contained in a half sphere.  To see this, consider the edge $T_{4}\cap T_{3}$.  Let $\Pi_{12}$ be the plane such that $\Pi_{12}\cap S^{2}\supset T_{4}\cap T_{3}$.  Temporarily suppose $\theta_{13}=\theta_{23}=\pi/2$, and let $R\colon\R^{3}\to\R^{3}$ be the origin fixing rotation that maps the midpoint of $T_{4}\cap T_{3}$ to $v_{3}$.  Then $T_{4}$ is contained in a quarter sphere whose boundary is contained in $\Pi_{12}\cap S^{2}$ and $(R(\Pi_{12}))\cap S^{2}$.  If $\theta_{13},\theta_{23}\leq\pi/2$, then $T_{4}$ remains in this quarter sphere, since otherwise $v_{3}\notin\{x\in S^{2}\colon d_{S^{2}}(x,v_{1})\leq\pi/2,d_{S^{2}}(x,v_{2})\leq\pi/2\}$.  So, Proposition \ref{prop5}(a) is violated.  All cases therefore produce a contradiction, concluding the proof that $\theta_{12}>1/10$.

We now prove that $\sin\Theta_{ijk}>1/10$.  We argue by contradiction and assume that $\sin\Theta_{132}<1/10$.  The procedure is similar to the calculations above.  The theme is: if some $\Theta_{ijk}$ is very small or very large, this creates additive relations among the edges, up to small errors.  Before we begin, we list a few consequences of the spherical law of cosines, for $i,j,k$ distinct elements of $\{1,2,3,4\}$.

\begin{equation}\label{four4.7}
\cos\theta_{ij}=\cos(\theta_{ik}-\theta_{kj})+\sin\theta_{ik}\sin\theta_{kj}(\cos\Theta_{ikj}-1).
\end{equation}
\begin{equation}\label{four4.8}
\cos\theta_{ij}=\cos(\theta_{ik}+\theta_{kj})+\sin\theta_{ik}\sin\theta_{kj}(\cos\Theta_{ikj}+1).
\end{equation}

Now, assume $\Theta_{132}<\pi/2$.  Since $\Theta_{132}+\Theta_{234}+\Theta_{134}=2\pi$ and $-\Theta_{134}>-\pi$ from Proposition \ref{prop1}(e), $\Theta_{132}+\Theta_{234}\geq\pi$.  So $\Theta_{234}>\pi/2$ and $\sin\Theta_{234}\leq\sin(\pi-\Theta_{132})=\sin(\Theta_{132})<1/10$.  So, by relabeling edges if necessary, we may assume that $\Theta_{132}>\pi/2$.  Since $\Theta_{132}>\pi/2$, $\cos\Theta_{132}\leq-.9949$.  Without loss of generality, label the edges so that $\theta_{23}=\max\{\theta_{13},\theta_{23}\}$ and $\theta_{13}=\min\{\theta_{13},\theta_{23}\}$.

\noindent{\bf (Case 1')}  Assume $\max\{\theta_{13},\theta_{23}\}<\pi/3$.  By \eqref{one6}, $\theta_{34}\leq\theta_{13}+\theta_{23}\leq2\pi/3$.  But then Proposition \ref{prop5}(a) is violated.  (Let $v\in T_{1}\cap T_{2}$ satisfy $d_{S^{2}}(v,v_{4})=\min\{\theta_{34},\pi/2\}$, and note that $\{v_{i}\}_{i=1}^{4}\subset\{x\in S^{2}\colon d(x,v)\leq\pi/2\}$.)  So Case 1' cannot occur.

\noindent{\bf (Case 2')}  Assume $\max\{\theta_{13},\theta_{23}\}\geq\pi/3$ and $1/10<\min\{\theta_{13},\theta_{23}\}\leq1/2$.  Then $\theta_{13}+\theta_{23}<\pi$ by Lemma \ref{lemma4}, so \eqref{four4.5} and \eqref{four4.6} give $\abs{\theta_{12}-(\theta_{13}+\theta_{23})}<.11$.  By the spherical law of sines and Lemma \ref{lemma4}, $\sin\Theta_{123}=\sin\Theta_{132}\sin\theta_{13}/\sin\theta_{12}\leq\sin\Theta_{132}\sin(1/2)/\sin(1/2)<1/10$.
Since $\theta_{13}\leq\theta_{23}$ and $\theta_{13}\leq1/2$, Lemma \ref{lemma4} gives $\cos\theta_{13}\geq\abs{\cos\theta_{23}}\geq\cos\theta_{23}\cos\theta_{12}$.  So, the formula $\cos\Theta_{123}=(\cos\theta_{13}-\cos\theta_{23}\cos\theta_{12})/(\sin\theta_{12}\sin\theta_{23})$ shows $\cos\Theta_{123}\geq0$, i.e. $\cos\Theta_{123}>.9949$.

Since $\theta_{23}<\pi-1/2$ from Lemma \ref{lemma4}, $\theta_{13}+\theta_{23}<\pi$.  Combining this with $\cos\Theta_{132}<-.9949$, \eqref{four4.8} and \eqref{four4.6}, we see that $\theta_{12}\geq\theta_{13}+\theta_{23}-.11$.  Since $\cos\Theta_{123}>.9949$, \eqref{one6} and $\theta_{12}>.1$ imply $\theta_{24}\geq\sqrt{(\pi/3)^{2}+(.1)^{2}+2(.1)(\pi/3)(.9949)}>\pi/3$.  Now, since $\cos\Theta_{123}>.9949$, $\Theta_{123}+\Theta_{124}+\Theta_{324}=2\pi$, and $-\Theta_{123}>-\pi$ from Proposition \ref{prop1}(e), $\Theta_{123}+\Theta_{324}\geq\pi$, so $\cos\Theta_{324}\leq\cos(\pi-\Theta_{123})<-.9949$.  Since $\theta_{23},\theta_{24}>\pi/3$, by relabeling the edges if necessary we may assume that $\cos\Theta_{123}<-.9949$, $\max\{\theta_{13},\theta_{23}\}\geq\pi/3$ and $1/2<\pi/3<\min\{\theta_{13},\theta_{23}\}$.  Thus, we enter into the following case.

\noindent{\bf (Case 3')}  Assume $\max\{\theta_{13},\theta_{23}\}\geq\pi/3$ and $1/2<\min\{\theta_{13},\theta_{23}\}$.  Recall that we may label the edges so that $\theta_{23}=\max\{\theta_{13},\theta_{23}\}$ and $\theta_{13}=\min\{\theta_{13},\theta_{23}\}$.  We claim that $\theta_{23}>1.3$.  If not, then \eqref{one6} says $\theta_{34}^{2}=\theta_{23}^{2}+\theta_{13}^{2}+2\theta_{23}\theta_{13}\cos\Theta_{132}$, so $\theta_{34}\leq\sqrt{2(1.3)^{2}-2(.994)(\pi/3)(1/2)}<1.53<\pi/2$ and $\{v_{i}\}_{i=1}^{4}\subset\{x\in S^{2}\colon d_{S^{2}}(x,v_{3})\leq\pi/2\}$.  Thus, Proposition \ref{prop5}(a) is violated, and our claim is proven.

As in Case 1 in the proof that $\theta_{ij}>1/10$, \eqref{four4.8} gives a dichotomy
$$(i)\quad\abs{\theta_{12}-(\theta_{13}+\theta_{23})}<.11, \qquad\mbox{or}\quad(ii)\quad \abs{\theta_{12}+\theta_{13}+\theta_{23}-2\pi}<.11$$
In case $(i)$, $\theta_{12}\geq\theta_{23}+\theta_{13}-.11>1/2$, and in case $(ii)$, Lemma \ref{lemma4} says $\theta_{ij}<\pi-1/2$, so $\theta_{12}\geq.89>1/2$.  So in either case, $1/2\leq\theta_{12}\leq\pi-1/2$, using Lemma \ref{lemma4} again.  By the spherical law of sines, $\sin\Theta_{123}\leq\sin\Theta_{132}/\sin\theta_{12}\leq(1/10)(1/\sin(1/2))<21/100$.  Similarly, $\sin\Theta_{312}<21/100$.

Assume $(i)$ holds.  If $\Theta_{123}<\pi/2$, then $\cos\Theta_{123}>.977$.  By $(i)$, $\theta_{12}\geq1.3+1/2-.11=1.69$.  By \eqref{one6}, $\theta_{24}^{2}=\theta_{12}^{2}+\theta_{23}^{2}+2\theta_{12}\theta_{23}\cos\Theta_{123}$, so
$$
\theta_{24}\geq\sqrt{(1.69)^{2}+(\pi/3)^{2}+2(1.69)(\pi/3)(.977)}>2.72>\pi-\frac12.
$$
contradicting Lemma \ref{lemma4}.  If $\Theta_{123}\geq\pi/2$ then $\cos\Theta_{123}<-.977$.  Then \eqref{four4.8} gives a dichotomy: either $(i)'$ $\abs{\theta_{13}-(\theta_{12}+\theta_{23})}<.26$ or $(ii)'$ $\abs{\theta_{13}+\theta_{12}+\theta_{23}-2\pi}<.26$.  If $(i)'$ holds, then $\theta_{13}\geq\theta_{12}+\theta_{23}-.26$, but $(i)$ says $\theta_{12}\geq\theta_{13}+\theta_{23}-.11$.  So $\theta_{13}\geq\theta_{13}+2\theta_{23}-.37$, i.e. $\theta_{23}<.185$, contradicting that $\theta_{23}>\pi/3$.  So $(ii)'$ must hold.  But then $(i)$ implies $2\theta_{12}\geq\theta_{12}+\theta_{13}+\theta_{23}-.11\geq2\pi-.37$, so $\theta_{12}\geq\pi-.185$, contradicting Lemma \ref{lemma4}.  Therefore $(i)$ does not hold.

Assume $(ii)$ holds.  Suppose $\Theta_{123}<\pi/2$ so $\cos\Theta_{123}>.977$.  Then \eqref{four4.7} and \eqref{four4.6} show that $\abs{\theta_{13}-\abs{\theta_{12}-\theta_{23}}}<.26$.  If $\theta_{12}\geq\theta_{23}$, then $\abs{\theta_{12}-(\theta_{13}+\theta_{23})}<.26$.  So by $(ii)$ $2\theta_{12}\geq\theta_{12}+\theta_{13}+\theta_{23}-.26\geq2\pi-.37$, contradicting Lemma \ref{lemma4}.  If $\theta_{23}\geq\theta_{12}$, then the same argument shows $2\theta_{23}\geq2\pi-.37$, contradicting Lemma \ref{lemma4}.  We may therefore assume that $\cos\Theta_{123}<-.977$.  If $\cos\Theta_{312}>.977$, we get the same contradiction, so we may assume $\cos\Theta_{312}<-.977$.  In summary, $\cos\Theta_{123},\cos\Theta_{312},\cos\Theta_{132}<-.977$.

For $i,j$ distinct, $i,j\in\{1,2,3,4\}$ let $E_{ij}=T_{i}\cap T_{j}$.  Let $i',j'$ so that $\theta_{i'j'}\colonequals\max\{\theta_{13},\theta_{23},\theta_{12}\}$.  Let $\Pi_{i'j'}$ be the plane containing $E_{i'j'}$ and the origin, and let $n$ be the unit normal to $\Pi_{i'j'}$ such that $\langle n,z_{4}\rangle<0$.  We claim that the edges $E_{i'4},E_{j'4}$ satisfy
\begin{equation}\label{four4.93}
E_{i'4},E_{j'4}\subset\{x\in S^{2}\colon0\leq\langle x,n\rangle\leq.65\}.
\end{equation}
Also, the union of $E_{i'4},E_{j'4},E_{i'j'}$ forms a noncontractible loop in this topological annulus.

To prove $E_{i'4}\subset\{x\in S^{2}\colon\langle x,n\rangle\leq.65\}$, it suffices to show
\begin{equation}\label{four4.95}
\Theta_{j'i'4}>\pi-\sin^{-1}(.65),\quad\mbox{or}\quad\theta_{i'4}<\sin^{-1}(.65).
\end{equation}
To see this consequence, let $x\in S^{2}$ be contained in the great circle containing $E_{i'4}$ and let $\{i',j',k\}=\{1,2,3\}$.  Since $\langle n,v_{i'}\rangle=0$ and $\langle n,z_{k}\rangle>0$, $\langle x,n\rangle$ is maximized when $d_{S^{2}}(x,v_{i'})=\pi/2$ and $\langle x,n\rangle\geq0$.  For such an $x$, $\langle x,n\rangle=\sin\Theta_{j'i'4}$.  So, in the case that $\Theta_{j'i'4}>\pi-\sin^{-1}(.65)$, $E_{i'4}\subset\{x\in S^{2}\colon\langle x,n\rangle\leq.65\}$.  Let now $x\in\{y\in S^{2}\colon d_{S^{2}}(y,v_{i'})\leq\theta_{i'4}\}\supset E_{i'4}$.  For such an $x$, $\langle x,n\rangle$ is maximized for $x$ in the plane containing $\{v_{i'},n,0\}$ with $d_{S^{2}}(x,v_{i'})=\theta_{i'4}$.  So $\langle x,n\rangle=\sin\theta_{i'4}$.  So, in the case that $\theta_{i'4}<\sin^{-1}(.65)$, $E_{i'4}\subset\{x\in S^{2}\colon\langle x,n\rangle\leq.65\}$, as desired.  The containment $E_{j'4}\subset\{x\in S^{2}\colon\langle x,n\rangle\leq.65\}$ follows similarly if one shows: $\Theta_{i'j'4}>\pi-\sin^{-1}(.65)$ or $\theta_{j'4}<\sin^{-1}(.65)$.  Specifically, in the above paragraph we switch all indices of the form $i'$ to $j'$ and all $j'$ to $i'$, and the containment follows.

We now discuss the proof of \eqref{four4.95}.  Let $\{i',j',k\}=\{1,2,3\}$.  Since $\Theta_{j'i'k},\Theta_{i'j'k}<\pi$ by Proposition \ref{prop1}(e) and $E_{i'j'}\subset\{x\in S^{2}\colon\langle x,n\rangle=0\}$, we conclude that $E_{i'k},E_{j'k}\subset\{x\in S^{2}\colon\langle x,n\rangle\leq0\}$.  By Lemma \ref{lemma5}, $\sum_{j\colon j\neq i'}\exp_{i'}^{-1}(j)=0$ and $\sum_{j\colon j\neq j'}\exp_{j'}^{-1}(j)=0$, so we conclude $E_{i'4},E_{j'4}\subset\{x\in S^{2}\colon0\leq\langle x,n\rangle\}$.  By the definition of $\theta_{i'j'}$ and Lemma \ref{lemma5}, $\Theta_{j'i'4}\geq\Theta_{ki'4}$.  In particular, since $\Theta_{j'i'4}+\Theta_{ki'4}\geq\pi$, we conclude that $\Theta_{j'i'4}\geq\pi/2$.  Similarly, $\Theta_{i'j'4}\geq\pi/2$.  If $E_{i'4}\cup E_{i'j'}\cup E_{j'4}$ were contractible in $\{x\in S^{2}\colon0\leq\langle x,n\rangle\leq.65\}$, then by \eqref{four4.95}, the location of $v_{4}$ is such that $T_{3}$ avoids the region $\{x\in S^{2}\colon\langle x,n\rangle\geq.65\mbox{ or }\langle x,n\rangle\leq0\}$.  But then either $\Theta_{i'j'4}$ or $\Theta_{j'i'4}$ is less than $\pi/2$, a contradiction.  So, the noncontractibility property of $E_{i'4}\cup E_{i'j'}\cup E_{j'4}$ follows from \eqref{four4.95}, which remains to be proven.

Consider $\theta_{i'k}$ and $E_{i'j'}$ as fixed with $\Theta_{j'i'k},E_{i'4}$ variable.  From Lemma \ref{lemma5}, recall that $\exp_{i'}^{-1}(4)$ and $\exp_{i'}^{-1}(j')$ have angle $\Theta_{j'i'4}$.  Also $\exp_{i'}^{-1}(k)$ and $\exp_{i'}^{-1}(j')$ have angle $\Theta_{j'i'k}$, and $\exp_{i'}^{-1}(\cdot)$ is confined to a plane.  By Lemma \ref{lemma5}, scaling $\theta_{i'j'}$ and $\theta_{i'k}$ by the same constant (greater than one) leaves $\Theta_{j'i'4}$ unchanged and increases $\theta_{i'4}$.  So in proving \eqref{four4.95}, we may assume that $\theta_{i'j'}=\pi-1/2$, by Lemma \ref{lemma4}.

Suppose $\theta_{i'k}<.815(\pi-1/2)$.  In this case, we will show that $\sin\Theta_{j'i'4}<.65$, recalling that $\Theta_{j'i'4}>\pi/2$.  Using planar geometry and Lemma \ref{lemma5}, one can show: decreasing $\Theta_{j'i'k}$ decreases $\Theta_{j'i'4}$.
So, for the purpose of showing $\sin\Theta_{j'i'4}<.65$, we may assume that $\cos\Theta_{j'i'k}=-.977$.

Now, suppose $\theta_{i'k}\geq.815(\pi-1/2)$.  In this case  we show that $\sin\theta_{i'4}<.65$, recalling that $\theta_{i'4}<\pi/2$, since \eqref{one7.2} says
$$\theta_{i'4}^{2}
=(\theta_{i'j'}-\theta_{i'k})^{2}+2\theta_{i'j'}\theta_{i'k}(1+\cos\Theta_{j'i'k})
\leq(.185(\pi-1/2))^{2}+2(\pi-1/2)^{2}(.815)(.033)<.62.$$
Consider $\theta_{i'k},E_{i'j'}$ as fixed with $\Theta_{j'i'k},E_{i'4}$ variable.  By Lemma \ref{lemma5},
$$
\theta_{i'4}^{2}=\vnorm{\exp_{i'}^{-1}(4)}_{2}^{2}
=\vnorm{\exp_{i'}^{-1}(j')+\exp_{i'}^{-1}(k)}_{2}^{2}
=\theta_{i'j'}^{2}+\theta_{i'k}^{2}+2\theta_{i'j'}\theta_{i'k}\cos\Theta_{j'i'k}.
$$
Recall $\cos\Theta_{j'i'k}<-.977$, so decreasing $\Theta_{j'i'k}$ increases $\theta_{i'4}$.  So, for the purpose of maximizing $\sin\theta_{i'4}$, we may assume that $\cos\Theta_{j'i'k}=-.977$.  With these reductions (i.e. that $\cos\Theta_{j'i'k}=-.977$ and $\theta_{i'j'}=\pi-1/2$) for these two cases, one can then verify \eqref{four4.95} directly as a one-dimensional inequality, using Lemma \ref{lemma5} and treating $\theta_{ik}$ as a variable with all other quantities a function of $\theta_{ik}$.  With identical reductions, $\Theta_{i'j'4}>\pi-\sin^{-1}(.65)$ $\theta_{j'4}<\sin^{-1}(.65)$.  Therefore, \eqref{four4.93} holds.
%
%
%

Given \eqref{four4.93}, note that the circumference of the geodesic ball $\{x\in S^{2}\colon \langle x,n\rangle=.65\}$ is given by $2\pi\sqrt{1-(.65)^2}>4.77$.  By the noncontractible property, $\theta_{i'j'}+\theta_{i'4}+\theta_{j'4}\geq4.77$.  By \eqref{one6}, $\theta_{i'4}\geq\sqrt{(\theta_{i'j'}-\theta_{i'k})^{2}+2\theta_{i'j'}\theta_{i'k}(1+\cos\Theta_{j'i'k})}\geq\theta_{i'j'}-\theta_{i'k}$ and similarly $\theta_{j'4}\geq(\theta_{i'j'}-\theta_{j'k})$.  Substituting these inequalities into $\theta_{i'j'}+\theta_{i'4}+\theta_{j'4}\geq4.77$ gives
$$
\theta_{i'j'}+(\theta_{i'j'}-\theta_{i'k})+(\theta_{i'j'}-\theta_{j'k})\geq4.77.
$$
However, by $(ii)$, $\theta_{i'j'}+\theta_{i'k}+\theta_{j'k}\geq2\pi-.11$.  So adding the inequalities gives $4\theta_{i'j'}\geq4.77+2\pi-.11$, i.e. $\theta_{i'j'}>2.73>\pi-1/2$, violating Lemma \ref{lemma4}.  So, all cases produce a contradiction, and the result is proven.
\end{proof}

\section{Numerical Computations}
\label{secnet}

We now provide a more comprehensive analysis of the system of equations resulting from Lemma \ref{prop7}.  We begin by writing this system explicitly.  For $i\in\{1,2,3\}$ choose $j,\ell$ such that $\{i,j,\ell\}=\{1,2,3\}$.  Define

\begin{equation}\label{three0}
\gamma_{i}\colonequals-\sqrt{\lambda}\,\langle v_{i},v_{4}\rangle
=\frac{\sin\theta_{i\ell}}{\theta_{i\ell}}\cos\theta_{ij}+\frac{\sin\theta_{j\ell}}{\theta_{j\ell}}+\frac{\sin\theta_{ij}}{\theta_{ij}}\cos\theta_{i\ell}.
\end{equation}
Then $\gamma=\gamma_{2}$ from \eqref{one3.0}.  Our system follows by cyclically permuting the $\theta_{ij}$ in \eqref{one5}, i.e.,
\begin{equation}\label{three1}
\begin{cases}
\sqrt{\lambda}\cos\left(\sqrt{\theta_{12}^{2}+\theta_{13}^{2}+2\theta_{12}\theta_{13}\cos\Theta_{213}}\right)+\gamma_{1}=0,\\
\sqrt{\lambda}\cos\left(\sqrt{\theta_{12}^{2}+\theta_{23}^{2}+2\theta_{12}\theta_{23}\cos\Theta_{123}}\right)+\gamma_{2}=0,\\
\sqrt{\lambda}\cos\left(\sqrt{\theta_{13}^{2}+\theta_{23}^{2}+2\theta_{13}\theta_{23}\cos\Theta_{132}}\right)+\gamma_{3}=0,
\end{cases}
\end{equation}
where $\Theta_{ij\ell}$ and $\lambda$ are defined in~\eqref{eq:def thetas} and \eqref{one1.6} respectively.

Observe that \eqref{three1} gives three equations in the three unknowns $\theta_{12},\theta_{13},\theta_{23}$.  Define $H\colon[0,\pi]^{3}\to\R^{3}$ so that $H(\theta_{12},\theta_{13},\theta_{23})=(H_{1},H_{2},H_{3})$ is the vector corresponding to the entries of the left side of \eqref{three1}.  We now wish to find $(\theta_{12},\theta_{13},\theta_{23})$ such that $H=(0,0,0)$.

We first examine neighborhoods of the two points discussed in Section~\ref{sec:identities}.

\begin{lemma}\label{lemma7}
The global maximum of $F$ does not occur in the $\ell_{2}$ balls of radius $1/100$ around the points $\{\theta_{12},\theta_{13},\theta_{23}\}=\{\theta,\theta,\theta\}\in\R^{3}$ with $\theta=\arccos(-1/3)$ and $\theta=1.53796841207904$.
\end{lemma}
\begin{proof}
Let $\delta=1/100$.  We claim that $\abs{\langle\nabla_{\R^{3}}F_{0}(\theta_{12},\theta_{13},\theta_{23}),v\rangle}\leq20$, for $(\theta_{12},\theta_{13},\theta_{23})$ in an $\ell_{2}$ ball in $\R^{3}$ of radius $\delta$ ball around the point $(\arccos(-1/3),\arccos(-1/3),\arccos(-1/3))\in\R^{3}$, where $v\in S^{2}$. Now, $F_{0}(\arccos(-1/3),\arccos(-1/3),\arccos(-1/3))=6(\arccos(-1/3))^{2}\approx21.9031$ while $F$ at the spherical propeller partition  evaluates to $(9/4)\pi^{2}\approx 22.206$.  So, given this claim, since $\delta\abs{\langle\nabla_{\R^{3}}F_{0}\rangle}<22.206-21.903$ in $B_{2}^{3}(\{\theta\}_{i=1}^{3},\delta)$, no global maximum of $F$ may occur in $B_{2}^{3}(\{\theta\}_{i=1}^{3},\delta)$.  Similarly, we claim that $\abs{\langle\nabla_{\R^{3}}F_{0}(\theta_{12},\theta_{13},\theta_{23}),v\rangle}\leq15$ for $(\theta_{12},\theta_{13},\theta_{23})$ in an $\ell_{2}$ ball of radius $\delta$ around $\{\theta\}_{i=1}^{3}\in\R^{3}$ with $\theta\approx1.53796841207904$.  Since $F_{0}(\{\theta\}_{i=1}^{3})\approx21.7391$, once again, no global maximum of $F$ occurs in $B_{2}^{3}(\{\theta\}_{i=1}^{3},\delta)$.

Proof of claim: for $\theta=\arccos(-1/3)$ or $\theta=1.537968$, we first show that
$\vnorm{\mbox{Hessian}(F_{0})}_{2}<488$ for all points in $B_{2}^{3}(\{\theta\}_{i=1}^{3},\delta)\subset\R^{3}$.  Let $\tilde{F}(x,y,z)\colonequals F_{0}(x,y,z)$.  Then
\begin{flalign*}
\tilde{F}_{xx}
&=
6
+4\,{\frac {\cos \left( y \right) y}{\sin \left( y \right) }}
-2\,{\frac {\cos \left( y \right) xy\cos \left( x \right) }{\sin \left( x \right) \sin \left( y \right) }}
+4\,{\frac { \left( \cos \left( z \right) -\cos \left( x \right) \cos \left( y \right)  \right) xy \left( \cos \left( x \right)  \right) ^{2}}{ \left( \sin \left( x \right)  \right) ^{3}\sin \left( y \right) }}\\
&-4\,{\frac { \left( \cos \left( z \right) -\cos \left( x \right) \cos \left( y \right)  \right) y\cos \left( x \right) }{ \left( \sin \left( x \right)  \right) ^{2}\sin \left( y \right) }}
+2\,{\frac { \left( \cos \left( z \right) -\cos \left( x \right) \cos \left( y \right)  \right) xy}{\sin \left( x \right) \sin \left( y \right) }}-2\,{\frac {\cos \left( x \right) zy}{\sin \left( z \right) \sin \left( y \right) }}\\
&+4\,{\frac {\cos \left( z \right) z}{\sin \left( z \right) }}
-2\,{\frac {\cos \left( z \right) xz\cos \left( x \right) }{\sin \left( x \right) \sin \left( z \right) }}
+4\,{\frac { \left( \cos \left( y \right) -\cos \left( x \right) \cos \left( z \right)  \right) xz \left( \cos \left( x \right)  \right) ^{2}}{ \left( \sin \left( x \right)  \right) ^{3}\sin \left( z \right) }}\\
&-4\,{\frac { \left( \cos \left( y \right) -\cos \left( x \right) \cos \left( z \right)  \right) z\cos \left( x \right) }{ \left( \sin \left( x \right)  \right) ^{2}\sin \left( z \right) }}
+2\,{\frac { \left( \cos \left( y \right) -\cos \left( x \right) \cos \left( z \right)  \right) xz}{\sin \left( x \right) \sin \left( z \right) }}.
\end{flalign*}
\begin{flalign*}
\tilde{F}_{xy}
&=
-2\,xy
-2\,{\frac { \left( \cos \left( y \right)  \right) ^{2}xy}{ \left( \sin \left( y \right)  \right) ^{2}}}
+2\,{\frac {\cos \left( y \right) x}{\sin \left( y \right) }}
-2\,{\frac { \left( \cos \left( x \right)  \right) ^{2}xy}{ \left( \sin \left( x \right)  \right) ^{2}}}\\
&+2\,{\frac { \left( \cos \left( z \right) -\cos \left( x \right) \cos \left( y \right)  \right) xy\cos \left( x \right) \cos \left( y \right) }{ \left( \sin \left( x \right)  \right) ^{2} \left( \sin \left( y \right)  \right) ^{2}}}
-2\,{\frac { \left( \cos \left( z \right) -\cos \left( x \right) \cos \left( y \right)  \right) x\cos \left( x \right) }{ \left( \sin \left( x \right)  \right) ^{2}\sin \left( y \right) }}\\
&+2\,{\frac {\cos \left( x \right) y}{\sin \left( x \right) }}
-2\,{\frac { \left( \cos \left( z \right) -\cos \left( x \right) \cos \left( y \right)  \right) y\cos \left( y \right) }{\sin \left( x \right)  \left( \sin \left( y \right)  \right) ^{2}}}
+2\,{\frac {\cos \left( z \right) -\cos \left( x \right) \cos \left( y \right) }{\sin \left( x \right) \sin \left( y \right) }}\\
&+2\,{\frac {\sin \left( x \right) zy\cos \left( y \right) }{\sin \left( z \right)  \left( \sin \left( y \right)  \right) ^{2}}}
-2\,{\frac {\sin \left( x \right) z}{\sin \left( z \right) \sin \left( y \right) }}
+2\,{\frac {\sin \left( y \right) xz\cos \left( x \right) }{ \left( \sin \left( x \right)  \right) ^{2}\sin \left( z \right) }}
-2\,{\frac {\sin \left( y \right) z}{\sin \left( x \right) \sin \left( z \right) }}.
\end{flalign*}
With $\delta<1/100$, $\theta=\arccos(-1/3)$, and $(x,y,z)\in B_{2}^{3}(\{\theta\}_{i=1}^{3},\delta)$, we have $1/\abs{\sin(x)}<1.1$, and similarly for $y,z$.  Therefore, in this $\delta$-ball we have $\absf{\tilde{F}_{xx}}<228$.  Here we have bounded the cosine terms by $1$, the $x,y,z$ terms by $2$ (since $\arccos(1/3)+\delta<2$), and the inverted sine terms by $1.1$.  Using these same bounds, we find that $\absf{\tilde{F}_{xy}}<117$ in this same $\delta$-ball.  For a $3\times 3$ matrix $A=a_{ij}$, let $\vnorm{A}_{2}\colonequals(\sum_{i,j=1}^{3}a_{ij}^{2})^{1/2}$.  Since $\tilde{F}$ is symmetric in its arguments, we have the following bound in $B_{2}^{3}(\{\theta\}_{i=1}^{3},\delta)$
$$\vnormf{\mbox{Hessian}(\tilde{F})}_{2}
\leq\vnorm{\begin{pmatrix} 228 & 117 & 117\\ 117 & 228 & 117\\ 117 & 117 & 228\end{pmatrix}}_{2}
<488,$$
as desired.  And the same estimate of the Hessian applies for $\theta\approx 1.5379684120790425$.

Apply now the Hessian estimate to $\vnorm{\nabla_{\R^{3}} f(b)}_{2}\leq\vnorm{\nabla_{\R^{3}} f(a)}_{2}+\delta\sup_{\vec{x}\in B_{2}^{3}(a,\delta)}\vnormf{\mbox{Hessian}(f)(\vec{x})}_{2}$, which holds for every sufficiently smooth $f\colon\R^{3}\to\R$ if $a,b\in \R^3$ satisfy $\vnorm{b-a}_{2}\leq\delta$.  We then explicitly calculate $\vnormf{\nabla_{\R^{3}}F_{0}|_{\{\arccos(-1/3)\}_{i=1}^{3}}}_{2}<13.6$ and $\vnormf{\nabla_{\R^{3}}F_{0}|_{\{1.5379684\}_{i=1}^{3}}}_{2}<9.7$.  So, in $B_{2}^{3}(\{\arccos(-1/3)\}_{i=1}^{3},\delta)$ we have $\vnorm{\nabla_{\R^{3}}F_{0}}_{2}<13.6+4.88$ and in $B_{2}^{3}(\{1.5379684\}_{i=1}^{3},\delta)$ we have $\vnorm{\nabla_{\R^{3}}F_{0}}_{2}<9.7+4.88$, proving the claim.
\end{proof}

\begin{lemma}\label{rk5}
We may assume that $\sqrt{\lambda}>.18$.
\end{lemma}
\begin{proof}
Substituting in \eqref{one1.5}, $\det(v_{2},v_{1},v_{4})=\det(v_{1},v_{2},v_{3})\frac{\sin\theta_{12}}{\theta_{12}}\frac{1}{\sqrt{\lambda}}$.  Using \eqref{three2}, \begin{equation}\label{three1.01}
\abs{\det(v_{2},v_{1},v_{4})}>\abs{\det(v_{1},v_{2},v_{3})}\frac{\sin(1/2)}{\sqrt{\lambda}\,(\pi-1/2)}.
\end{equation}
Now, suppose $\sqrt{\lambda}<\sin(1/2)/(\pi-1/2)$, and $\sqrt{\lambda'}<\sin(1/2)/(\pi-1/2)$, with $\lambda'$ defined by
$$\lambda'
\colonequals\vnorm{\frac{\sin\theta_{12}}{\theta_{12}}v_{4}+\frac{\sin\theta_{24}}{\theta_{24}}v_{1}+\frac{\sin\theta_{14}}{\theta_{14}}v_{2}}_{2}^{2}
$$
Applying \eqref{three1.01} twice (with indices changed appropriately) shows $\abs{\det(v_{2},v_{1},v_{4})}>\abs{\det(v_{1},v_{2},v_{3})}$ and $\abs{\det(v_{2},v_{1},v_{4})}<\abs{\det(v_{1},v_{2},v_{3})}$, a contradiction.  We may therefore assume that one triangle of $P$ (which, up to relabeling, is $T_{4}$) satisfies $\sqrt{\lambda}>(\sin(1/2)/(\pi-1/2))>.18$.
\end{proof}

We now derive some modulus of continuity and derivative estimates for our system \eqref{three1}.  Recall that $P=\{T_{i}\}_{i=1}^{4}$ satisfies Proposition \ref{prop1}(e).  Since \eqref{three1} must hold for the data of each individual triangle $T_{i}$, we select in \eqref{three1} the data from the triangle $T_{4}$.  Then $\theta_{12}+\theta_{23}+\theta_{13}\leq 2\pi$, since the perimeter of $T_{4}$ is bounded by $2\pi$.  By relabeling the edges of $T_{4}$ and using the pigeonhole principle, we may assume $\theta_{12}\leq2\pi/3$. The following summary of~\eqref{three1} holds for any labeling of the form $\{i,j,\ell\}=\{1,2,3\}$.
\begin{equation}\label{three1.1}
\begin{cases}
\theta_{ij}+\theta_{j\ell}>\theta_{i\ell} &\mbox{, triangle inequality on the sphere}\\
\theta_{12}+\theta_{23}+\theta_{13}\leq 2\pi\\
\theta_{12}\leq2\pi/3\\
\theta_{13},\theta_{23}\leq\pi-1/2 & \mbox{, \eqref{three2}}\\
\max_{i,j\in\{1,2,3\},i\neq j}\{\theta_{ij}\}\geq\pi\frac{3}{2\sqrt{14}} & \mbox{, Corollary \ref{cor2.5}}\\
\theta_{ij}\geq1/10 & \mbox{, Lemma \ref{thm1}}\\
\sin\Theta_{ij\ell}\geq1/10 & \mbox{, Lemma \ref{thm1}}\\
\sqrt{\lambda}>.18 & \mbox{, Lemma \ref{rk5}}
\end{cases}
\end{equation}

For a map $f\colon (X,\vnorm{\cdot}_{X})\to (Y,\vnorm{\cdot}_{Y})$ between two normed linear spaces, a modulus of continuity for $f$, denoted by $\omega_{f}$, satisfies $\vnorm{f(x_{1})-f({x_{2})}}_{Y}\leq\omega_{f}(\vnorm{x_{1}-x_{2}}_{X},x_{1},x_{2})$.  We allow $\omega_{f}$ to depend on $x_{1}$, with $x_{2}$ in a $\delta$-ball around some given point $x_{1}=(\theta_{12},\theta_{13},\theta_{23})$.

\begin{lemma}\label{lemma6}
Suppose $\delta<1/100$, and consider $H_{1}\colon([0,\pi]^{3},\vnorm{\cdot}_{\infty})\to\R$, $H_{1}=H_{1}(\theta_{12},\theta_{13},\theta_{23})$.  Suppose also that $\lambda=\lambda(\theta_{12},\theta_{13},\theta_{23})>\eta$.  Then, in an $\ell_{\infty}$ ball of radius $\delta$ centered at $(\theta_{12},\theta_{13},\theta_{23})$, the following holds
\begin{equation}\label{three3}
\begin{aligned}
\omega_{H_{1}}(\delta)
&\leq G(\delta,\vec{\theta},\eta)\colonequals\delta\left[\frac{7}{2}+\frac{15}{2\sqrt{\eta-15\delta}}\right.
+3\bigg(2(\theta_{12}+\delta)+ 2(\theta_{13}+\delta)\\
+&\left.(\theta_{12}+\delta)(\theta_{13}+\delta)\left(\frac{2}{\sin(\theta_{12}+\delta)}+\frac{2}{\sin(\theta_{13}+\delta)}
+\frac{1}{\sin(\theta_{12}+\delta)\sin(\theta_{13}+\delta)}\right)
\right)\bigg].
\end{aligned}
\end{equation}

Or, applying Lemma \ref{rk5}, for $\delta$ sufficiently small,

\begin{equation}\label{three5}
\begin{aligned}
\omega_{H_{1}}(\delta)
&\leq\delta\left[\frac{7}{2}+\frac{15}{2\sqrt{.0329-15\delta}}\right.
+3\bigg(2(\theta_{12}+\delta)+ 2(\theta_{13}+\delta)\\
+&\left.(\theta_{12}+\delta)(\theta_{13}+\delta)\left(\frac{2}{\sin(\theta_{12}+\delta)}+\frac{2}{\sin(\theta_{13}+\delta)}
+\frac{1}{\sin(\theta_{12}+\delta)\sin(\theta_{13}+\delta)}\right)
\right)\bigg].
\end{aligned}
\end{equation}
\end{lemma}
\begin{proof}
We want to compute the modulus for the following function from $(\R^{3},\vnorm{\cdot}_{\infty})$ to $\R$
$$
H_{1}(\theta_{12},\theta_{13},\theta_{23})
=\sqrt{\lambda}\cos\left(\sqrt{\theta_{12}^{2}+\theta_{13}^{2}+2\theta_{12}\theta_{13}\cos\Theta_{213}}\right)+\gamma_{1}
$$
where $\lambda$ and $\gamma_{1}$ are defined in \eqref{one1.7} and \eqref{three0} respectively.  In what follows, we use the product rule, chain rule, and sum rule for the modulus $\omega_{(\cdot)}$.
\begin{flalign*}
\omega_{H_{1}}(\delta)
& \leq\omega_{\sqrt{\lambda}}(\delta)
+\vnormf{\sqrt{\lambda}}_{\infty}\cdot\omega_{\cos\left(\sqrt{\theta_{12}^{2}+\theta_{13}^{2}+2\theta_{12}\theta_{13}\cos\Theta_{213}}\right)}(\delta)
+\omega_{\gamma_{1}}(\delta)\\
&\leq\omega_{\sqrt{\lambda}}(\delta)+3\cdot\frac{1}{2}\cdot\omega_{(\theta_{12}^{2}+\theta_{13}^{2}+2\theta_{12}\theta_{13}\cos\Theta_{213})}(\delta)
+\frac{7}{2}\delta.
\end{flalign*}
Here we used $\lambda\leq9$, $\sqrt{\lambda}\leq3$, $\abs{\frac{d}{dt}(\cos(\sqrt{t}))}\leq1/2$, $\abs{(\sin t)/t}\leq1$, and $\absf{\frac{d}{dt}((\sin t)/t)\leq1/2}$.  Also, the term $(7/2)\delta$ appears since $\delta^{-1}\omega_{\gamma_{1}}(\delta)\leq\vnorm{\nabla_{\R^{3}}\gamma_{1}(\theta_{12},\theta_{13},\theta_{23})}_{1}\leq 7/2$.  Now, observe that
\begin{flalign*}
\frac{\partial \lambda}{\partial\theta_{12}}
&=2\frac{\sin\theta_{12}}{\theta_{12}}\left(\left.\frac{d}{dx}\frac{\sin x}{x}\right|_{x=\theta_{12}}\right)
+2\frac{\sin\theta_{13}\sin\theta_{23}}{\theta_{13}\theta_{23}}(-\sin\theta_{12})\\
&+2\left(\left.\frac{d}{dx}\frac{\sin x}{x}\right|_{x=\theta_{12}}\right)\frac{\sin\theta_{13}}{\theta_{13}}\cos\theta_{23}
+2\left(\left.\frac{d}{dx}\frac{\sin x}{x}\right|_{x=\theta_{12}}\right)\frac{\sin\theta_{23}}{\theta_{23}}\cos\theta_{13}.
\end{flalign*}
So, taking absolute values, and using that $\abs{\sin(t)/t}\leq1$ and $\abs{\frac{d}{dt}(\sin(t)/t)}\leq1/2$, we see that $\abs{\partial \lambda/\partial\theta_{ij}}\leq5$.  So, if $\lambda(\vec{\theta})=\lambda(\theta_{12},\theta_{13},\theta_{23})>\eta>0$ for some $\eta$, and if $\vec{x}=(x,y,z)$ satisfies $\abs{x-\theta_{12}}<\delta,\abs{y-\theta_{13}}<\delta,\abs{z-\theta_{23}}<\delta$, then at $\vec{x}$,
$$\abs{\frac{\partial(\sqrt{\lambda})}{\partial\theta_{12}}}
\leq\max_{\left\{\vec{\theta}\in\R^{3}\colon\vnormf{\vec{\theta}-\vec{x}}_{\infty}<\delta\right\}}
\frac{1}{2}\frac{1}{\sqrt{\lambda(\vec{\theta})}}\cdot\frac{\partial\lambda(\vec{\theta})}{\partial\theta_{12}}
\leq\frac{1}{2}\frac{1}{\sqrt{\eta-\vnorm{\nabla_{\R^{3}}\lambda}_{1}\delta}}\cdot5
\leq\frac{1}{2}\frac{1}{\sqrt{\eta-15\delta}}\cdot5.$$
Therefore, for such $\vec{x}$,
$$\vnormf{\nabla_{\R^{3}}(\sqrt{\lambda})}_{1}\leq\frac{15}{2\sqrt{\eta-15\delta}}.$$

Combining this estimate with our estimate above, we get the following modulus estimate, that applies to $\vec{x},\vec{\theta}$ as just described
\begin{flalign*}
\omega_{H_{1}}(\delta)
&\leq\delta\left[\frac{7}{2}+\frac{15}{2\sqrt{\eta-15\delta}}\right.
+3\bigg(2(\theta_{12}+\delta)+ 2(\theta_{13}+\delta)\\
+&\left.(\theta_{12}+\delta)(\theta_{13}+\delta)\left(\frac{2}{\sin(\theta_{12}+\delta)}+\frac{2}{\sin(\theta_{13}+\delta)}
+\frac{1}{\sin(\theta_{12}+\delta)\sin(\theta_{13}+\delta)}\right)
\right)\bigg].
\end{flalign*}

We used here the monotonicity of the function $x\mapsto x/\sin x$ and the following identities for $\psi(\theta_{12},\theta_{13},\theta_{23})\colonequals (\theta_{12}^{2}+\theta_{13}^{2}+2\frac{\theta_{12}\theta_{13}}{\sin\theta_{12}\sin\theta_{13}}(\cos\theta_{23}-\cos\theta_{12}\cos\theta_{13}))$,
\begin{equation*}
\begin{aligned}
\frac{1}{2}\frac{\partial \psi}{\partial \theta_{12}}
&=\theta_{12}+\frac{\theta_{12}\theta_{13}}{\sin\theta_{12}\sin\theta_{13}}(\sin\theta_{12}\cos\theta_{13})\\
&\qquad+(\cos\theta_{23}-\cos\theta_{12}\cos\theta_{13})(\theta_{12}/\sin\theta_{12})'\cdot(\theta_{13}/\sin\theta_{13})\\
&=\theta_{12}+\frac{\theta_{12}\theta_{13}}{\sin\theta_{13}}\cos\theta_{13}
+\cos\Theta_{213}\left(1-\frac{\theta_{12}}{\tan\theta_{12}}\right)\theta_{13},\\
\frac{1}{2}\frac{\partial \psi}{\partial \theta_{13}}
& =\theta_{13}+\frac{\theta_{12}\theta_{13}}{\sin\theta_{12}\sin\theta_{13}}(\sin\theta_{13}\cos\theta_{12})\\
&\qquad+(\cos\theta_{23}-\cos\theta_{12}\cos\theta_{13})(\theta_{13}/\sin\theta_{13})'\cdot(\theta_{12}/\sin\theta_{12})\\
&=\theta_{13}+\frac{\theta_{12}\theta_{13}}{\sin\theta_{12}}\cos\theta_{12}
+\cos\Theta_{213}\left(1-\frac{\theta_{13}}{\tan\theta_{13}}\right)\theta_{12},\\
\frac{1}{2}\frac{\partial \psi}{\partial\theta_{23}}
&=\frac{\theta_{12}\theta_{13}}{\sin\theta_{12}\sin\theta_{13}}(-\sin\theta_{23}).
\end{aligned}
\end{equation*}
\end{proof}

Analogous estimates give a gradient bound for $F_{0}$, as defined in \eqref{one8}.  Since $F_{0}$ is symmetric in its arguments, it suffices to take a derivative with respect to one variable.  Fix $(\theta_{12},\theta_{13},\theta_{23})$, and let $0<\delta<1/100$.  For any $(x,y,z)$ such that $\abs{x-\theta_{12}}<\delta,\abs{y-\theta_{13}}<\delta,\abs{z-\theta_{23}}<\delta$ we have the following derivative bound

\begin{equation}\label{three7}
\begin{aligned}
\abs{\frac{\partial F_{0}}{\partial\theta_{12}}}
&\leq 6(\theta_{12}+\delta)
+\frac{2(\theta_{13}+\delta)(\theta_{23}+\delta)}{\sin(\theta_{13}+\delta)\sin(\theta_{23}+\delta)}
+2\frac{(\theta_{12}+\delta)(\theta_{13}+\delta)}{\sin(\theta_{13}+\delta)}\\
&+2\frac{(\theta_{12}+\delta)(\theta_{23}+\delta)}{\sin(\theta_{23}+\delta)}
+2(\theta_{13}+\theta_{23}+2\delta)\left(1-\frac{\theta_{12}+\delta}{\tan(\theta_{12}+\delta)}\right)
\end{aligned}
\end{equation}


The restrictions of \eqref{three1.1} and the bounds of \eqref{three3} and \eqref{three7} finally yield Theorem \ref{thm:main}.

\begin{proof}[Proof of Theorem \ref{thm:main}]
Let $P$ be a partition of the sphere $S^{2}$ that maximizes $\sum_{i=1}^{4}\vnormf{z_{i}}_{2}^{2}$.  Arguing by contradiction, we may assume that this partition exceeds the value of the spherical propeller partition, $9\pi^{2}/4$.  From Proposition \ref{prop1}  we may assume that $P$ is a collection of four nonempty spherical triangles along with four vertices such that $\det(v_{1},v_{2},v_{3})>0$.  Since $P$ is a global maximum of $F$, it satisfies three equations in the three variables $\{\theta_{12},\theta_{13},\theta_{23}\}$ defined by \eqref{three1}.  We may restrict these variables via \eqref{three1.1}.  We may also exclude neighborhoods of the two candidate maxima found in Section~\ref{sec:identities}, using Lemma \ref{lemma7}.  Combining the estimates of \eqref{three3},\eqref{three5} and \eqref{three7} gives sufficient information for an $\epsilon$-net traversal of the parameter space defined by \eqref{three1.1}.  This numerical computation finds that no zeros of the system \eqref{three1}, within the parameter space defined by \eqref{three1.1}, exceed the objective value of the propeller partition.  This computation contradicts the assumed existence of a maximal $P$ with four nonempty elements.  So, recalling \eqref{eq:def z_i}, $\sum_{i=1}^{4}\vnorm{z_{i}}_{2}^{2}\leq9\pi^{2}/4=2\pi^{3}(9/(8\pi))$.
\end{proof}

\begin{subsection}{Implementation}\label{sec:implement}

We describe the $\epsilon$-net traversal of Theorem \ref{thm:main} further below.  Before doing so, we describe some issues related to the numerical verification of inequalities.  In particular, we discuss how a computer can rigorously verify a mathematical statement.  Recall that a normal double precision floating point number in the IEEE-754-2008 standard is a number in base two of the form
$$\pm1.a_{1}a_{2}a_{3}\cdots a_{52}\times 2^{b_{1}b_{2}\cdots b_{11}-2^{10}},$$
where $a_{i}\in\{0,1\}$ and $b_{1}b_{2}\cdots b_{11}\in[2,2^{11}-1]\cap\Z$.  All of our computations are done with this discrete set of numbers, with zero included as well. Note that the distance between two consecutive numbers changes according to the values of the two numbers.  For example, the distance between $1$ and the next largest number is $2^{-52}$, whereas the distance between $2$ and the next largest number is $2^{-51}$.  (In general, if the number $2^{k}$ and the next largest number can be represented in this system, then their distance is $2^{-52+k}$.)  Due to this spacing between numbers, one must round the result of any arithmetic operation.  For example, the addition $1+2^{-54}$ evaluates to $1$, if we choose to round to the nearest double precision number.  For the sake of flexibility, our analysis below applies regardless of the rounding scheme that is chosen.


For $x,y\in[0,2^{500}]$ where $x,y$ are normal double precision numbers, let $\mbox{fl}(x+y)$ denote the normal double precision number that is evaluated in a computer with the IEEE-754-2008 standard.  It suffices for now to note that $(x+y)\leq(1+2\cdot2^{-52})\mbox{fl}(x+y)$.  In this case, we say that $(x+y)$ has multiplicative error bounded by $(1+2\cdot2^{-52})$.  When we speak of numerical error, we mean multiplicative error, unless otherwise stated.  Define
$$\varepsilon\colonequals 3\cdot2^{-52}.$$
To verify inequalities in our numerical computations, we perform operations for several floating point numbers and keep track of error bounds of the form $(x+y)\leq(1+2\cdot2^{-52})\mbox{fl}(x+y)$.  Since $(1+2\cdot2^{-52})^{N}\leq(1+N\cdot3\cdot 2^{-52})=(1+N\varepsilon)$ for $1\leq N\leq30000$, we can bound the multiplicative error of $N$ operations with a multiplicative term of the form $(1+N\varepsilon)$.  (Here and below we are overly conservative in our error estimates.)

We now present the formulas that are used in our computations, delaying further discussion of numerical errors.  Let $x,y>0$ be two normal double precision floating point numbers with $x>y$.  A loss of significance refers to a loss of binary significant digits in the computation of $\mbox{fl}(x-y)$.  One can show that $\mbox{fl}(x-y)$ loses at most $q$ and at least $p$ significant binary digits, where $2^{-q}\leq 1-(y/x)\leq 2^{-p}$. Due to potential loss of significance errors, we evaluate $\lambda$ as
\begin{equation}\label{one1.7prime}
\begin{split}
\lambda & = \left(\frac{\sin^{2}\theta_{23}}{\theta_{23}^{2}}+\frac{\sin^{2}\theta_{13}}{\theta_{13}^{2}}+\frac{\sin^{2}\theta_{12}}{\theta_{12}^{2}}\right)
+2\left(\frac{\sin\theta_{23}\sin\theta_{13}}{\theta_{23}\theta_{13}}\cos\theta_{12}\right.\\
&\qquad+\left.\frac{\sin\theta_{23}\sin\theta_{12}}{\theta_{23}\theta_{12}}\cos\theta_{13}
+\frac{\sin\theta_{12}\sin\theta_{13}}{\theta_{12}\theta_{13}}\cos\theta_{23}\right).
\end{split}
\end{equation}

Our numerical computation of inequalities consists of four steps.

\noindent{\bf (Step (i))} Suppose we want to check \eqref{three1.1} for a closed $\ell_{\infty}$ ball (i.e. a cube) of radius $\delta$ with center $(\theta_{12},\theta_{13},\theta_{23})\in[0,\pi]^{3}$.  To do so, we must verify the following perturbation of \eqref{three1.1}.
\begin{equation}\label{three1.1prime}
\begin{cases}
(\theta_{ij}+\theta_{j\ell}+3\delta)\geq\theta_{i\ell}\\
\theta_{12}+\theta_{23}+\theta_{13}\leq (2\pi+3\delta)\\
\theta_{12}\leq(2\pi/3+\delta)(1+10\varepsilon)\\
\theta_{13}+1/2,\theta_{23}+1/2\leq(\pi+\delta)(1+10\varepsilon)\\
(\delta+\max_{i,j\in\{1,2,3\},i\neq j}\{\theta_{ij}\})(1+7\varepsilon)\geq\pi\frac{3}{2\sqrt{14}}(1-5\varepsilon)\\
(\theta_{ij}+\delta)(1+7\varepsilon)\geq(1/10)(1-\varepsilon)\\
-(3\delta+\frac{\sqrt{99}}{10}(\sin\theta_{ij}\sin\theta_{j\ell}+2\delta))(1+1500\varepsilon)
\leq(\cos\theta_{i\ell}-\cos\theta_{ij}\cos\theta_{j\ell}+2000\varepsilon)\\
(\cos\theta_{i\ell}-\cos\theta_{ij}\cos\theta_{j\ell}-2000\varepsilon)
\leq(3\delta+\frac{\sqrt{99}}{10}(\sin\theta_{ij}\sin\theta_{j\ell}+2\delta))(1+1500\varepsilon)\\
(15\delta+2\cdot10^{4}\varepsilon+\lambda)\geq.03293
\end{cases}
\end{equation}
If a given inequality of \eqref{three1.1prime} does not hold, then the corresponding inequality of \eqref{three1.1} is violated for the entire $\ell_{\infty}$ ball of radius $\delta$.

\noindent{\bf (Step (ii))} If all of the inequalities of \eqref{three1.1prime} are satisfied, we then check if $(\theta_{12},\theta_{13},\theta_{23})$ lies in a neighborhood of two of our known candidate global optima.  Let $\vec{x}=(x_{1},x_{2},x_{3})\in[0,\pi]^{3}$.  Suppose $x_{1}=x_{2}=x_{3}=\arccos(-1/3)$ or $x_{1}=x_{2}=x_{3}=1.5379684120790425$.  From Lemma \ref{lemma7}, we know that if $\vec{\theta}=(\theta_{12},\theta_{13},\theta_{23})$ is close to $\vec{x}$, then $F$ is bounded by $(9/4)\pi^{2}$ in $\overline{B_{\infty}^{3}}((\theta_{12},\theta_{13},\theta_{23}),\delta)$.  That is, if the following is not satisfied, then $F$ is bounded as just stated.
\begin{equation}\label{five0}
\left[(d_{\ell_{2}}(\vec{\theta},\vec{x}))^{2}+40000\varepsilon
+2\delta\left(d_{\ell_{1}}(\vec{\theta},\vec{x})+6000\varepsilon+3\delta\right)\right](1+20\varepsilon)
>((1/100)^{2})(1-2\varepsilon).
\end{equation}
Here $d_{\ell_{2}}$ and $d_{\ell_{1}}$ denote the usual $\ell_{2}$ and $\ell_{1}$ metrics on $\R^{3}$.

\noindent{\bf (Step (iii))} If the verification of \eqref{five0} fails, we use the modulus estimates of Lemma \ref{lemma6} to see if $H$ is near zero at $\vec{\theta}$, with $H=(H_{1},H_{2},H_{3})$ defined after \eqref{three1}.  Let $\alpha,\beta,\rho,\sigma\in\{\pm1\}$.  We first check if the following expression has constant sign, over the $2^{4}$ such choices of $\alpha,\beta,\rho,\sigma$.
\begin{equation}\label{five1.0}
\begin{aligned}
&\left[
\cos\left(\sqrt{\theta_{12}^{2}+\theta_{13}^{2}+2\theta_{12}\theta_{13}
\left(\frac{\cos\theta_{23}-\cos\theta_{12}\cos\theta_{13}+2000\varepsilon\alpha}{\sin\theta_{12}\sin\theta_{13}}\right)}\,\right)+25000\varepsilon\beta
\right]
\\
&\quad\cdot\sqrt{\lambda-2\cdot10^{4}\varepsilon\sigma}+(\gamma_{1}+3000\varepsilon\rho).
\end{aligned}
\end{equation}
If \eqref{five1.0} is well-defined and its sign is constant, we then consider the following inequalities involving $H_{i}$ for $i=1$ and \eqref{three3}.
\begin{equation}\label{five1}
\begin{aligned}
\min_{\alpha,\beta,\rho,\sigma\in\{\pm1\}}
&\left|\left[
\cos\left(\sqrt{\theta_{12}^{2}+\theta_{13}^{2}+2\theta_{12}\theta_{13}
\left(\frac{\cos\theta_{23}-\cos\theta_{12}\cos\theta_{13}+2000\varepsilon\alpha}{\sin\theta_{12}\sin\theta_{13}}\right)}\,\right)+25000\varepsilon\beta
\right]\right.
\\
&\cdot\sqrt{\lambda-2\cdot10^{4}\varepsilon\sigma}+(\gamma_{1}+3000\varepsilon\rho)\Bigg|(1-3000\varepsilon)
>
(1+3000\varepsilon)G(\delta,\vec{\theta},\lambda-2\cdot10^{4}\varepsilon).
\end{aligned}
\end{equation}
\begin{equation}\label{five1.1}
\begin{aligned}
&\min_{\alpha,\beta,\rho,\sigma\in\{\pm1\}}
\left|\left[
\cos\left(\sqrt{\theta_{12}^{2}+\theta_{13}^{2}+2\theta_{12}\theta_{13}
\left(\frac{\cos\theta_{23}-\cos\theta_{12}\cos\theta_{13}+2000\varepsilon\alpha}{\sin\theta_{12}\sin\theta_{13}}\right)}\,\right)+25000\varepsilon\beta
\right]\right.
\\
&\quad\cdot\sqrt{\lambda-2\cdot10^{4}\varepsilon\sigma}+(\gamma_{1}+3000\varepsilon\rho)\Bigg|(1-3000\varepsilon)\\
&>
(1+5\varepsilon)\sqrt{15\delta}+\delta(1+3000\varepsilon)\left[\frac{7}{2}\right.
+3\bigg(2(\theta_{12}+\delta)+ 2(\theta_{13}+\delta)
\\
&\quad+\left.(\theta_{12}+\delta)(\theta_{13}+\delta)\left(\frac{2}{\sin(\theta_{12}+\delta)}+\frac{2}{\sin(\theta_{13}+\delta)}
+\frac{1}{\sin(\theta_{12}+\delta)\sin(\theta_{13}+\delta)}\right)
\right)\bigg].
\end{aligned}
\end{equation}
The inequalities for $H_{2}$ and $H_{3}$ are constructed by permuting cyclically the indices appearing above.  If at least one of \eqref{five1} and \eqref{five1.1} is satisfied (for at least one $H_{i}$), then $H_{i}$ is nonzero on $\overline{B_{\infty}^{3}}((\theta_{12},\theta_{13},\theta_{23}),\delta)$.

\noindent{\bf (Step (iv))}  If \eqref{five1} and \eqref{five1.1} are not satisfied for $\{H_{i}\}_{i=1}^{3}$, we finally check the value of $F_{0}$ directly, using a modification of \eqref{one8}.  Define
\begin{flalign*}
\cos\Theta_{ijk}^{\varepsilon}
&\colonequals\left(\frac{\cos\theta_{ik}-\cos\theta_{ij}\cos\theta_{jk}+2000\varepsilon}{\sin\theta_{ij}\sin\theta_{jk}}\right)\\
F_{0}^{\varepsilon}
&\colonequals3\left(\theta_{12}^{2}+\theta_{23}^{2}+\theta_{13}^{2}\right)
+2\cos(\Theta_{213}^{\varepsilon})\theta_{12}\theta_{13}
+2\cos(\Theta_{123}^{\varepsilon})\theta_{12}\theta_{23}
+2\cos(\Theta_{231}^{\varepsilon})\theta_{23}\theta_{13},
\\
G_{12}^{\varepsilon}(\vec{\theta},\delta)
&\colonequals6(\theta_{12}+\delta)
+\frac{4(\theta_{13}+\delta)(\theta_{23}+\delta)}{\sin(\theta_{13}+\delta)\sin(\theta_{23}+\delta)}
+2\frac{(\theta_{12}+\delta)(\theta_{13}+\delta)}{\sin(\theta_{13}+\delta)}
\\
&\quad+2\frac{(\theta_{12}+\delta)(\theta_{23}+\delta)}{\sin(\theta_{23}+\delta)}
+2(\theta_{13}+\theta_{23}+2\delta)\left[1-\left(\frac{\theta_{12}+\delta}{\tan(\theta_{12}+\delta)}(1+1200\varepsilon)+\varepsilon\right)\right].
\end{flalign*}
If the following is satisfied, then $F_{0}$ is bounded by $(9/4)\pi^{2}$ in $\overline{B_{\infty}^{3}}((\theta_{12},\theta_{13},\theta_{23}),\delta)$, and therefore $F$ satisfies the same bound for any global maximum of $F$ in this ball.
\begin{equation}\label{five2}
(9/4)\pi^{2}(1-5\varepsilon)-F_{0}^{\varepsilon}\cdot(1+5000\varepsilon)
>\delta(G_{12}^{\varepsilon}+G^{\varepsilon}_{13}+G^{\varepsilon}_{23})(1+10^{4}\varepsilon)+100\varepsilon.
\end{equation}

Given these estimates, we can now describe the $\epsilon$-net traversal of Theorem \ref{thm:main}.  We need a discrete subset $\mathcal{N}=\{x_{i}\}\subset[0,\pi]^{3}$ together with a set of radii $\delta_{i}>0$ such that
$$\bigcup_{x_{i}\in\mathcal{N}}\overline{B_{\infty}^{3}}(x_{i},\delta_{i})\supset[0,\pi]^{3},$$
and such that one of the following cases occurs, for $(x_{i},\delta_{i})$ with $x_{i}\in\mathcal{N}$.
\begin{itemize}
\item[(i)] One inequality of \eqref{three1.1prime} does not hold.
\item[(ii)] \eqref{five0} does not hold.
\item[(iii)] \eqref{five1} or \eqref{five1.1} holds, and \eqref{five1.0} has constant sign over all sign choices $\alpha,\beta,\sigma,\rho\in\{\pm1\}$.
\item[(iv)] \eqref{five2} holds.
\end{itemize}
In case (i), $\overline{B_{\infty}^{3}}(x_{i},\delta_{i})$ is either outside of the domain of \eqref{three1} or \eqref{three1} cannot be satisfied in $\overline{B_{\infty}^{3}}(x_{i},\delta_{i})$.  Case (ii) implies that $F=\sum_{i=1}^{4}\vnorm{z_{i}}_{2}^{2}$ is bounded by $(9/4)\pi^{2}$ on $\overline{B_{\infty}^{3}}(x_{i},\delta_{i})$.  Cases (iii) and (iv) imply that no global maximum of $F$ occurs in $\overline{B_{\infty}^{3}}(x_{i},\delta_{i})$.  So, the existence of $\mathcal{N}$ as defined above completes the proof of Theorem \ref{thm:main}.

To construct $\mathcal{N}$, we perform the following standard procedure.  Begin with the grid
$$\mathcal{N}_{0}\colonequals \frac{\pi}{100}([0,100]^{3}\cap\Z^{3})\subset[0,\pi]^{3}.$$
For each point $x\in\mathcal{N}_{0}$, check the inequalities discussed above with $\delta=\delta_{0}\colonequals\pi/200$.  If we discover that $x$ satisfies one of (i) through (iv), then eliminate $x$ from the set $\mathcal{N}_{0}$.  Next, include the element $(x,\delta)$ in the set $\mathcal{N}$.  Let $\mathcal{N}_{0}'$ be the set of all points $x\in\mathcal{N}_{0}$ that have not been included in $\mathcal{N}$.  We now ``refine'' the grid $\mathcal{N}_{0}'$ by a factor of $10$ and repeat the above procedure.  That is, we define
$$\mathcal{N}_{1}
\colonequals \left(\frac{\pi}{1000}([0,1000]^{3}\cap\Z^{3})\right)\cap
\left(\bigcup_{x_{i}\in\mathcal{N}_{0}'}\overline{B_{\infty}^{3}}(x_{i},\pi/200)\right).
$$
For each point $x\in\mathcal{N}_{1}$, we now check (i)-(iv) with $\delta=\delta_{1}\colonequals\pi/2000$, and so on.  For $j\geq1$, at the $j^{th}$ iteration of this algorithm, we define
$$\mathcal{N}_{j}
\colonequals \left(\frac{\pi}{10^{j+2}}([0,10^{j+2}]^{3}\cap\Z^{3})\right)\cap
\left(\bigcup_{x_{i}\in\mathcal{N}_{j-1}'}\overline{B_{\infty}^{3}}(x_{i},\pi/(2\cdot10^{j+1}))\right).
$$
We continue in this manner until the procedure terminates.  Upon termination at the $j^{th}$ step, $\mathcal{N}_{j}'=\emptyset$, so $\mathcal{N}$ satisfies our desired properties and proves Theorem \ref{thm:main}.

We now return to our discussion of numerical error.  To keep track of rounding errors resulting from subtractions, we use absolute error where appropriate and add the $\varepsilon$ terms that correspond to absolute numerical errors.  These additive terms are scaled according to the absolute value of the numbers being subtracted (due to the unequal spacing of floating point numbers described above).

In the description and analysis of (i) through (iv), $\delta$ terms come from modulus of continuity estimates, and $\varepsilon$ terms account for numerical error.  The latter is discussed further below.  The first two inequalities of \eqref{three1.1prime} contain no $\varepsilon$ terms, since they can be checked in exact arithmetic.  Let $x\in[0,\pi-1/2+2\times10^{-2}]$.  (We make this restriction due to the fourth inequality of \eqref{three1.1prime}.)  Let us discuss the numerical computation of the sine function.  If $x>\pi/2$, we replace $x$ with $\pi-x$, and we evaluate $\sin(\pi-x)$.  Thus, we only need to consider $x\in[0,\pi/2]$.  For such $x$, we first use a $19^{th}$ degree (ten term) Taylor expansion around zero, summing terms in ascending degree.  In perfect arithmetic, we get a multiplicative error of one plus
$$\abs{\frac{x^{21}}{(21!)\sin(x)}}
=\abs{\frac{x}{\sin(x)}\frac{x^{20}}{(21!)}}
\leq\frac{\pi}{2}\frac{(\pi/2)^{20}}{21!}<2.6\times 10^{-16}
<\frac{\varepsilon}{2}.$$
Since $x\in[0,\pi/2]$, there are no loss of significance errors in the summation of the series.  The addition of each term of the Taylor series involves one operation, and there are ten such terms.  Moreover, the $n^{th}$ term contains $n+1$ operations.  In total, one evaluation therefore has less than $210$ operations.  So, very conservatively, the multiplicative error of $\sin(x)$ is bounded by $(1+250\varepsilon)$, if $x$ is represented exactly.  If the argument $x$ has multiplicative error $(1+k\varepsilon)$, then this same analysis shows that the multiplicative error of $\sin(x)$ is bounded by $(1+(250+100k)\varepsilon)$, since the term $x$ appears $100$ times in the formula for $\sin$.  For example, since $\theta_{ij}$ has multiplicative error $(1+3\varepsilon)$, we conservatively bound the error of $\sin\theta_{ij}$ by $(1+550\varepsilon)$.

To compute $\cos(x)$ for $x\in[0,\pi-1/2+2\times10^{-2}]$, we instead compute $\sin(\pi/2-x)$ with the procedure outlined above.  Due to our representation of the points $\theta_{ij}$, we compute $\pi/2-x$ as a rational multiple of $\pi$, so no loss of significance errors occurs in this subtraction.  Hence, the error estimates for $\cos\theta_{ij}$ are exactly the ones already used for $\sin\theta_{ij}$.  However, when evaluating the cosine of more complicated expressions, we need to revert to absolute errors.

We will also need to use the square root function.  Here we use the usual Newton's Method iteration, which is well known to have multiplicative error bounded by $(1+\varepsilon)$.

We now describe the error terms appearing in \eqref{three1.1prime}, which are perturbations of \eqref{three1.1}.  The guiding principle is that we want to perturb each inequality so that it is more easily satisfied.  To see why this is done, suppose an inequality of \eqref{three1.1prime} is not satisfied.  Then, a modulus of continuity estimate and a numerical error estimate, used in the definition of \eqref{three1.1prime}, show that the corresponding inequality of \eqref{three1.1} is not satisfied in an $\ell_{\infty}$ ball of radius $\delta$.  For example, contrast the first inequality of \eqref{three1.1} with that of \eqref{three1.1prime}.  We need only add an additional $\delta$ term for each $\theta_{ij}$ term, since the modulus of continuity of the function $(x,y,z)\mapsto x+y+z$, viewed as a mapping from $(\R^{3},\vnorm{\cdot}_{\infty})\to\R$, is exactly $3\delta$.

So, if the inequality of \eqref{three1.1prime} is violated at a given point $\vec{\theta}=(\theta_{12},\theta_{13},\theta_{23})$, then the modulus of continuity bound shows that the corresponding inequality of \eqref{three1.1} is violated in the ball $\overline{B_{\infty}^{3}}(\vec{\theta},\delta)$.  The same considerations apply to the second inequality of \eqref{three1.1prime}.  Note that the first and second inequality are checked in exact arithmetic in our computer implementation, so we do not need to take into account numerical error.  As we have mentioned above, our $\theta_{ij}$ and $\delta$ are represented as numbers of the form $\pi(a/c)$ where $a$ and $c$ are integers, and it suffices to perform additions of the form $\pi((a+b)/c)$, and the integer addition $a+b$ can be performed with no error.

The third inequality of \eqref{three1.1prime} only requires a modulus of continuity term of $\delta$.  However, we must now take into account numerical errors.  We count ten numerical operations on positive numbers (including the error in representing $\pi$ as a floating point number three times), so in the worst case, the multiplicative error is bounded by $(1+10\varepsilon)$.  Similar considerations apply in the next three inequalities.  In the seventh and eighth inequalities, we use our analysis of the sine function, and recall that the multiplicative error for these operations is always bounded by $1+550\varepsilon$.  So, in total, the evaluation of the term $-(3\delta+\frac{\sqrt{99}}{10}(\sin\theta_{ij}\sin\theta_{j\ell}+2\delta))$
has multiplicative error bounded by $1+1500\varepsilon$.  Here we have also included $\delta$ terms corresponding to the moduli of continuity, as before.

Unfortunately, we cannot use multiplicative error to deal with the cosine terms, since loss of significance can occur in the evaluation of $(\cos\theta_{i\ell}-\cos\theta_{ij}\cos\theta_{j\ell})$.  However, since the cosine terms are bounded by $1$ in absolute value, the absolute error from each cosine evaluation is bounded by $550\varepsilon$.  Therefore the total absolute error from the right hand side of the seventh inequality is bounded by $2000\varepsilon$.  This bound applies similarly to the penultimate inequality of \eqref{three1.1prime}.  The final inequality of \eqref{three1.1prime} is similar to the previous cases.  To get the $\delta$ terms, we examine \eqref{one1.7prime} and use that $\abs{(\sin(x)/x)'}<1/2$ and $\abs{\sin(x)/x}\leq1$ for $x\in[0,\pi]$.  To get the $\epsilon$ terms, note that each term $\sin^{2}\theta_{ij}/\theta_{ij}^{2}$ has absolute error bounded by $1200\varepsilon$.  And each term $\sin\theta_{ij}\sin\theta_{jk}\cos\theta_{ik}/\theta_{ij}\theta_{jk}$ has absolute error bounded by $1700\varepsilon$.  The total absolute error of $\lambda$ is therefore bounded by $2\cdot10^{4}\varepsilon$.

We now treat \eqref{five0}.  First, recall that \eqref{one5.1} yielded our two known zeros of the system \eqref{three1}.  The first zero $\theta_{ij}=\theta=\arccos(-1/3)$ is known exactly (up to multiplicative error $(1\pm\varepsilon)$).  However, we can only find the second zero computationally.  For the purposes of our computations, we therefore show that this second zero is contained in an interval of the form $\theta_{ij}=\theta\in[x-2000\varepsilon,x+2000\varepsilon]$ where $x=1.5379684120790425$.  To see this, we take \eqref{one5.1} and write it as a quotient of the left side over the right side, and we wish to find where the quotient is $1$.  For $\theta\in[x-2000\varepsilon,x+2000\varepsilon]$, the multiplicative error involved in calculating this quotient is bounded by $1+3000\varepsilon$.  This follows by our analysis of the sine and cosine functions, and since this particular region of parameters avoids loss of significance errors.  We can then check numerically that, for $\theta=x-2000\varepsilon$, this quotient, multiplied by $(1+3000\varepsilon)$ is less than $1$, while for $\theta=x+2000\varepsilon$, the quotient, multiplied by $(1-3000\varepsilon)$, is greater than $1$.  So, by our numerical analysis and continuity of the given function, a solution of \eqref{one5.1} must lie in the interval $\theta\in[x-2000\varepsilon,x+2000\varepsilon]$.

We can now confront the errors of \eqref{five0}.  Once again, the guiding principle is to add error terms so that the inequality is more easily satisfied.  First, the term $d_{\ell_{2}}(\vec{\theta},\vec{x})^{2}$ is the sum of three squared differences.  Due to numerical error in the operations, and the uncertainty of the zero $\theta_{ij}\approx1.5379684120790425$, each squared difference has an absolute error bounded by $2\pi(2000\varepsilon)+100\varepsilon$.  (The factor of $2\pi$ comes from the squaring of the difference, and since $\abs{x}\leq\pi$.) So, in total, the absolute error from the term $d_{\ell_{2}}(\vec{\theta},\vec{x})^{2}$ is bounded by $40000\varepsilon$.  Since we want to check this inequality in an $\ell_{\infty}$ ball of radius $\delta$, we use the modulus of continuity of $d_{\ell_{2}}^{2}$, which is $2\delta d_{\ell_{1}}$.  We therefore add the term $2\delta d_{\ell_{1}}(\vec{\theta},\vec{x})$, and again add terms to this quantity to take into account its own absolute error and its own modulus of continuity.  The combined effect of summing the error terms then accumulates a multiplicative error bounded by $1+20\varepsilon$.

We now explain the error terms of \eqref{five1}.  From the analysis of the cosine, we know that $\cos\theta_{ij}$ for $\theta_{ij}\in[0,\pi-1/2+2\times10^{-2}]$ has an absolute error bounded by $550\varepsilon$.  (This follows since the multiplicative error is bounded by $(1+550\varepsilon)$, and $\abs{\cos(\cdot)}\leq1$.)  So, the expression $\cos\theta_{23}-\cos\theta_{12}\cos\theta_{13}$ has absolute error bounded by $2000\varepsilon$.  The remaining operations under the square root of \eqref{five1} all have multiplicative error bounded by $(1+1200\varepsilon)$.

When the cosine of the square root term is evaluated, we need to again take into account absolute errors (since, unlike in the $\cos\theta_{ij}$ terms, we must now account for loss of significance).  Since the argument of the cosine term is bounded by $\pi$, we can bound the absolute error of the square root term by $4000\varepsilon$.  To take the cosine of the square root term, we need to return to our Taylor series analysis. Note that our Taylor series has Lipschitz constant bounded by $6$.  (This is a conservative bound that does not account for cancelations in the series summation.)  Therefore, the total absolute error of the cosine of the square root in \eqref{five1} is bounded by $25000\varepsilon$.  From before, the absolute error in the computation of $\lambda$ in \eqref{one1.7prime} is bounded by $2\cdot10^{4}\varepsilon$.  Similarly, the computation of $\gamma_{1}$ has absolute error bounded by $3000\varepsilon$.

Temporarily ignoring any numerical errors, note that we are computing a bound for the absolute value of the derivative of a function.  So, in order to conclude anything from \eqref{five1}, the sign of \eqref{five1.0} must be constant over all possible numerical errors.  (If not, it may be that \eqref{five1.0} has an actual value of zero.)  Given that this sign is constant, we then take the minimum over all possible numerical errors.  It then suffices to take the minimum over the extreme points of a suitable rectangle, due to monotonicity of \eqref{five1} with respect to the error terms.  Note that we take the minimum since we wish to make the inequality more difficult to be satisfied.  Finally, the addition of all terms on the left side of the inequality incurs a multiplicative error bounded by $1+3000\varepsilon$, and the addition of the terms on the right side has multiplicative error bounded by $1+3000\varepsilon$.  A similar analysis applies to \eqref{five1.1}.

We now briefly mention the analysis of \eqref{five2}, which is routine at this stage.  The principle that we apply here is to make \eqref{five2} more difficult to be satisfied,  so we make its right side larger and its left side smaller.  The multiplicative error of the $F_{0}^{\varepsilon}$ term is bounded by $1+5000\varepsilon$, since each $\cos\Theta_{ijk}^{\varepsilon}$ term has multiplicative error bounded by $1+1200\varepsilon$.  Each $G_{ij}^{\varepsilon}$ term has multiplicative error bounded by $1+2500\varepsilon$, so the right side of \eqref{five2} has multiplicative error bounded by $1+10^{4}\varepsilon$.  A $100\varepsilon$ term is added to the right side of \eqref{five2} to account for the error of the subtraction on the left side.  The $1+1200\varepsilon$ term and the $\varepsilon$ term in the definition of $G_{ij}^{\varepsilon}$ account for the errors of the tangent function, and the loss of significance in the subtraction that occurs in the square brackets.


\end{subsection}

\section{Discrete Harmonic Maps into the Sphere}
\label{secconc}



We now summarize and reinterpret the results of Section \ref{sec:identities}.  Recall that we defined $G(P)=\sum_{1\leq i<j\leq4}\theta_{ij}^{2}$ as the sum of the squared lengths of the edges of our partition $P$.  We then saw from Propositions \ref{prop8} and \ref{prop1} that all maxima of $F$ are also critical points of $G$.  We also found that critical points of $G$ satisfy \eqref{one10prime}.  Given a partition $P$, think of its edges as rubber bands.  (Recall that a rubber band has energy given by its length squared, and it exerts a force proportional to its length.)  The critical points of $G$ are given by configurations of rubber bands (arranged like the edges of Figure \ref{fig1}(b)) that are at equilibrium with respect to the forces exerted by the rubber bands.  With this interpretation, \eqref{one10prime} becomes intuitively clear.  This equation says that, at a fixed vertex, the (tangential) forces of the incident rubber bands are in equilibrium.  (For the Euclidean version of this, see e.g. \cite{connelly82} and \cite{linial88}.)

These considerations lead us naturally to the following generalization.  Let $(V,E)$ be a (finite) graph.  Let $u\colon(V,E)\to S^{2}$ be an embedding of $(V,E)$ into $S^{2}$ ($u$ is sometimes called a map).  To be consistent with the above, let $u(V)=\{v_{i}\}$, and label $\mbox{Length}(u(\overline{v_{i}v_{j}}))=\theta_{ij}$ where $\overline{v_{i}v_{j}}$ denotes a spherical geodesic connecting $v_{i}$ to $v_{j}$.  Then a \textbf{harmonic map} $u$ is a critical point of the energy functional $\sum_{\{i,j\}\in E}\theta_{ij}^{2}$.

%
%
%

\bibliographystyle{abbrv}
\bibliography{propR3}

\end{document}